\documentclass{article}
\usepackage{fullpage}

\usepackage{mathtools, amssymb, amsthm}
\usepackage{xspace}
\usepackage{hyperref}
\usepackage{cleveref}

\usepackage{easyReview}
\usepackage{todonotes}

\newtheorem{theorem}{Theorem}
\newtheorem{definition}{Definition}
\newtheorem{lemma}{Lemma}
\newtheorem{claim}{Claim}[lemma]

\newenvironment{claimproof}[1][\proofname] 
{
    \proof[Proof of claim]
    
}
{
    \endproof
}

\theoremstyle{plain}
\newtheorem{rrule}{Reduction Rule}
\Crefname{rrule}{Reduction Rule}{Reduction Rules}

\newcommand{\NP}{\ensuremath{\mathsf{NP}}\xspace}

\renewcommand{\P}{\ensuremath{\mathsf{P}}\xspace}

\newcommand{\containment}{\ensuremath{\mathsf{NP \subseteq coNP/poly}}\xspace}

\newcommand{\N}{\mathbb{N}}
\newcommand{\Z}{\mathbb{Z}}

\newcommand{\Oh}{\mathcal{O}}

\newcommand{\tG}{\ensuremath{\tilde{G}}\xspace}
\newcommand{\side}{\texttt{side}}

\newcommand{\cI}{\ensuremath{\mathcal{I}}}

\DeclareMathOperator{\OPT}{OPT}

\newcommand{\conf}[2]{\textsc{conf}_{#1}(#2)}
\newcommand{\cut}{\textsc{cut}}

\newcommand{\probname}[1]{\lowercase{\textsc{#1}}}
\newcommand{\problem}[1]{\probname{#1}\xspace}
\newcommand{\parameterizedproblem}[2]{\textup{\probname{#1}[#2]}\xspace}

\newcommand{\vertexcover}{\problem{Vertex Cover}}
\newcommand{\VC}{\ensuremath{\text{VC}}}
\newcommand{\vertexcoveroct}{\parameterizedproblem{Vertex Cover}{oct}}

\newcommand{\oddcycletransversal}{\problem{Odd Cycle Transversal}}

\newcommand{\oddcycletransversallocal}{\parameterizedproblem{Odd Cycle Transversal}{local solution}}

\newcommand{\dtmultiwaycut}{\problem{Deletable Terminal Multiway Cut}}
\newcommand{\dtmultiwaycutlocal}{\parameterizedproblem{Deletable Terminal Multiway Cut}{local solution}}
\newcommand{\dtmultiwaycutterm}{\parameterizedproblem{Deletable Terminal Multiway Cut}{$|T_G|$}}

\newcommand{\smultiwaycut}{\problem{$s$-Multiway Cut}}
\newcommand{\smultiwaycutlocal}{\parameterizedproblem{$s$-Multiway Cut}{local solution}}

\title{Boundaried Kernelization via Representative Sets}
\author{
    Leonid Antipov\\Humboldt-Universität zu Berlin, Berlin, Germany\\leonid.antipov@hu-berlin.de 
    \and
    Stefan Kratsch\\Humboldt-Universität zu Berlin, Berlin, Germany\\stefan.kratsch@hu-berlin.de
}

\begin{document}
    
    \maketitle
    
    \begin{abstract}
        A kernelization is an efficient algorithm that given an instance of a parameterized problem returns an equivalent instance of size bounded by some function of the input parameter value. It is quite well understood which problems do or (conditionally) do not admit a kernelization where this size bound is polynomial, a so-called polynomial kernelization. Unfortunately, such polynomial kernelizations are known only in fairly restrictive settings where a small parameter value corresponds to a strong restriction on the global structure on the instance. Motivated by this, Antipov and Kratsch [WG 2025] proposed a local variant of kernelization, called boundaried kernelization, that requires only local structure to achieve a local improvement of the instance, which is in the spirit of protrusion replacement used in meta-kernelization [Bodlaender et al.\ JACM 2016]. They obtain polynomial boundaried kernelizations as well as (unconditional) lower bounds for several well-studied problems in kernelization.
        
        In this work, we leverage the matroid-based techniques of Kratsch and Wahlstr\"om [JACM 2020] to obtain randomized polynomial boundaried kernelizations for \smultiwaycut, \dtmultiwaycut, \oddcycletransversal, and \vertexcoveroct, for which randomized polynomial kernelizations in the usual sense were known before. A priori, these techniques rely on the global connectivity of the graph to identify reducible (irrelevant) vertices. Nevertheless, the separation of the local part by its boundary turns out to be sufficient for a local application of these methods.  
    \end{abstract}

    \section{Introduction}
    
    (Polynomial) kernelization is a very successful notion for rigorously studying efficient preprocessing for hard problems. A kernelization is an efficient algorithm that given an instance of a parameterized problem returns an equivalent instance of size bounded by some function of the input parameter value. It is a polynomial kernelization if this function is polynomially bounded. By now, it is quite well understood which problems do or (conditionally) do not admit a polynomial kernelization. In this way, we have learned what structural properties are helpful for a provable size reduction through efficient preprocessing, which is unlikely in general unless $\P=\NP$. Unfortunately, most polynomial kernelizations are for parameters that take low values only on instances with very restricted \emph{global structure}, e.g., only parameter many vertex deletions away from a known tractable special case of the problem. This is an issue because the size guarantee of a kernelization is only nontrivial when the parameter is small, otherwise we could just leave the instance as is.
    
    Motivated by this, Antipov and Kratsch~\cite{AntipovKratsch2025-BoundariedKernelization} proposed a local variant of kernelization, called \emph{boundaried kernelization}, inspired by protrusion replacement in meta-kernelization (Bodlaender et al.~\cite{DBLP:journals/jacm/BodlaenderFLPST16}). Roughly, this expects as input a boundaried graph $G_B$ and will return a boundaried graph $G'_B$ of size some function of the chosen parameter plus the boundary size (and possibly a shift $\Delta$ in solution value). The two graphs are equivalent in the sense that gluing any boundaried graph $H_B$ will result in equivalent instances (up to shift of $\Delta$ for optimization). Intuitively, such a preprocessing only needs a local part of a large graph to have beneficial structure and sufficiently small connection to the rest of the input, for the size bound to imply a size reduction. In this sense, boundaried kernelization relaxes the required structural conditions (from global to local), though of course also the size bound is only local (namely only for the boundaried local part). Since we can always run a boundaried kernelization with empty boundary, however, it is by itself in fact a strengthening of kernelization (modulo some technical aspects due to the variety of different parameterizations). Thus, it is natural and interesting to ask, which known polynomial kernelizations can be strengthened to polynomial boundaried kernelizations.
    
    In~\cite{AntipovKratsch2025-BoundariedKernelization}, polynomial boundaried kernelizations were obtained for \parameterizedproblem{Vertex Cover}{vc}, \parameterizedproblem{Vertex Cover}{fvs}, \parameterizedproblem{Feedback Vertex Set}{fvs}, \parameterizedproblem{Long Cycle}{vc}, \parameterizedproblem{Long Path}{vc}, \parameterizedproblem{Hamiltonian Cycle}{vc}, and \parameterizedproblem{Hamiltonian Path}{vc}.\footnote{Parameterized problems are denoted as problem name followed by parameter in brackets: These include parameterization by the size of a given vertex cover (vc), feedback vertex set (fvs), cluster vertex deletion set (cvd), cluster editing set (ce), and tree deletion set (tds), respectively.} In contrast, \parameterizedproblem{Cluster Editing}{cvd}, \parameterizedproblem{Cluster Editing}{ce}, \parameterizedproblem{Tree Deletion Set}{vc}, \parameterizedproblem{Tree Deletion Set}{tds}, and \parameterizedproblem{Dominating Set}{vc} were shown to (unconditionally) not admit a polynomial boundaried kernelization.\footnote{It was known that \parameterizedproblem{Dominating Set}{vc} has no polynomial kernelization unless \containment, unlike all other problems listed here, but exclusion of polynomial boundaried kernelization is unconditional.} Existence of unconditional lower bounds is not surprising and relies on exhibiting an unbounded (or at least too large) family of non-equivalent boundaried graphs. That being said, it is somewhat surprising that simple local reduction rules can sometimes be leveraged also for boundaried kernelization, while they otherwise fail completely.
    
    There are of course plenty of other graph problems with polynomial kernelizations to explore. Next to well-studied problems (as above) and cases with somewhat special kernelization (as \parameterizedproblem{Tree Deletion Set}{tds}), it would be interesting to know how inherently global tools fare in the boundaried setting. Here, the matroid-based techniques of Kratsch and Wahlstr\"om~\cite{KratschWahlstrom2020-RepresentativeSetsIrrelevantVertices} come to mind, as they rely on the properties of certain matroids defined on the entire graph, while also enabling the first (and so far only) polynomial kernelizations for certain cut and feedback problems. Another interesting one would be the randomized polynomial compression for \parameterizedproblem{Steiner Cycle}{$|T|$} due to Wahlstr\"om~\cite{DBLP:conf/stacs/Wahlstrom13} but likely one would first have to strengthen it to a kernelization, as it outputs a somewhat contrived matrix problem not known to be in \NP.
    
    \subparagraph{Our work.}
    Motivated by the question of whether the global matroid-based techniques of Kratsch and Wahlstr\"om~\cite{KratschWahlstrom2020-RepresentativeSetsIrrelevantVertices} can also be adapted for local preprocessing, we study the same problems in the boundaried setting. We are able to strengthen several results to polynomial boundaried kernelizations.
    
    \begin{theorem}\label{thm:main}
        The following parameterized problems admit randomized polynomial boundaried kernelization: \smultiwaycutlocal (with a minor restriction described below), \dtmultiwaycutterm, \oddcycletransversallocal, and \vertexcoveroct.
    \end{theorem}
    
    The main idea lies in the power of the gammoids and the resulting cut covers $Z$ which contain the vertices of an exponential number of min-cut queries. While Kratsch and Wahlstr\"om applied this on gammoids with sources being the terminal set, or an approximate odd cycle transversal, we additionally put in the boundary vertices as sources.
    As a result, this family of min-cut queries whose optimal answers are all covered by $Z$, contain in particular all possible combinations needed in order to complete an optimal global solution for any forced behavior of the boundary vertices in the solution. 
    For instance, in \smultiwaycut, where the solution is a vertex set $X$ such that each of the $s$ given terminals $t \in T$ is contained in its own component in $G-X$, we are able to preserve an optimal solution for every choice of putting any boundary vertex either into the solution, the component of some terminal, or some unrelated component. Or in \oddcycletransversal, where $G-X$ has to be a bipartite graph, the set $Z$ contains the local part of an optimal solution for every choice of putting $B$ vertices in the solution or one of the two parts of the bipartite graph $G-X$.
    
    Note however for the problem \smultiwaycut, that, for $G_B$ being the local graph, $T_G$ a set of terminals among $G$-vertices, and $H_B, T_H$ the rest of the global graph with the remaining terminals among the $H$-vertices, if we allow $T_H$ to also contain $B$-vertices, this forces us to be able to adapt to any pair of the boundary vertices to be terminals, and thus we need to store at least some information about the disjoint path sets between any such pair. Since the sizes of these path sets are not bounded, this means that we need to be able to store an unbounded amount of information, which directly rules out any effective preprocessing.
    
    \begin{theorem}\label{thm:side}
        In general, the parameterized problem \smultiwaycutlocal does not admit any boundaried kernelization.
    \end{theorem}
    
    So, the slightly restricted case that we are referring to in Theorem~\ref{thm:main} is that we forbid the terminal-part $T_H$ that is unknown for the kernelization, to contain any boundary vertices. This restriction is enough in order to obtain a boundaried kernelization, and it fits with the concept of local kernelization because the terminals in $B$ are known in that setting. (That being said, in general, it is appealing to push to the most general setting for getting a  polynomial boundaried kernelization. This includes having the preprocessing be fully independent of possible glued graphs $H_B$.)    
    
    \subparagraph{Related work.}
    Most immediately related is the already mentioned work on meta-kernelization via protrusion replacement (subgraphs with constant treewidth and constant boundary size). Inherently, the machinery for that does not lead to any polynomial size bounds, and that is not to be expected. Generally, the perspective of boundaried graphs is essential for dynamic programming on path and tree decompositions~(cf.~\cite{DBLP:books/sp/CyganFKLMPPS15}). In that setting, there is no hope for polynomial boundaried kernelizations because it would generalize the (likely) infeasible case of problems parameterized by path/treewidth regarding polynomial kernelization. Instead, our settings have more restrictive structure than small treewidth (as is usual for polynomial kernelization) and we obtain polynomial size bounds.
    
    Fomin et al.~\cite{DBLP:journals/talg/FominLST18} use a different form of protrusion for \parameterizedproblem{(Connected) Dominating Set}{$k$} where they require constant boundary size and a constant local solution size. It is used very differently than our (parameter) local cost, but is certainly reminiscent. Clearly, a small or constant local cost is a natural quality of interest amongst more explicit structural parameters like the size of modulators.
    
    Work by Jansen and Kratsch~\cite{DBLP:conf/esa/JansenK15} for preprocessing the feasibility problem of integer linear programs showed how to shrink subsystems with constant boundary size and either constant treewidth or total unimodularity to size polynomial in the domain. This predates the present notion of boundaried kernelization, like protrusion-replacement, but similarly does not yield size polynomial in the boundary size.    
    
    \subparagraph{Organization.}
    We give preliminaries regarding graph problems and boundaried kernelization in Section~\ref{section:prelim} and restate the tools from representative sets for matroids in Section~\ref{section:matroid}. In Sections~\ref{section:smwc}, \ref{section:dtmwc}, \ref{section:oct}, and \ref{section:vc} we describe the randomized polynomial boundaried kernelizations for \smultiwaycutlocal, \dtmultiwaycutterm, \oddcycletransversallocal, and \vertexcoveroct, respectively. Finally, we conclude in Section~\ref{section:conclusion}.
    
    \section{Preliminaries}\label{section:prelim}
    \subparagraph{Decision and optimization problems.}
    A \emph{decision problem} is a set $\Pi \subseteq \Sigma^*$. An \emph{optimization problem} is a function $\Pi \colon \Sigma^* \times \Sigma^* \to \N \cup \{\pm \infty\}$. For both we call any $x \in \Sigma^*$ an \emph{instance} of $\Pi$. For some instance $x$ and some $s \in \Sigma^*$, if $\Pi(x, s) \neq \pm \infty$, i.e., $\Pi(x, s) \in \N$, then $s$ is called a \emph{(feasible) solution} for $\Pi$ on $x$, and $\Pi(x, s)$ is called the \emph{value} of $s$. Optimization problems are divided into \emph{minimization} and \emph{maximization} problems. For a minimization (resp. maximization) problem we define the \emph{optimum value} for $\Pi$ on instance $x$, denoted $\OPT_\Pi(x)$ or simply $\OPT(x)$ if $\Pi$ is clear from the context, as the lowest (resp. highest) value of a feasible solution $s$ on $x$. If there is no feasible solution for $\Pi$ on $x$, then it holds that $\OPT_\Pi(x) = +\infty$ for a minimization problem and $\OPT_\Pi(x) = -\infty$ for a maximization problem. A solution with optimum value is called an \emph{optimum solution}.
    
    \subparagraph{Parameterized complexity.}
    A \emph{parameterized problem} is a set $Q \subseteq \Sigma^* \times \N$. For a tuple $(x, \ell) \in \Sigma^* \times \N$, which is again called an instance of $Q$, the number $\ell$ is called the \emph{parameter}. For a problem $\Pi$ (decision or optimization) and minimization problem $\rho$ we define the \emph{structurally parameterized} problem $\Pi[\rho]$, for which it holds that $(x, s, \ell) \in \Pi[\rho]$ if and only if $s$ is a feasible solution of value at most $\ell$ for $\rho$ on $x$ and: (i) for $\Pi$ being a decision problem, $x \in \Pi$; (ii) for $\Pi$ being a minimization problem, $x = (x', k)$ and $\OPT_\Pi(x') \leq k$; (iii) for $\Pi$ being a maximization problem, $x = (x', k)$ and $\OPT_\Pi(x') \geq k$. If $\Pi$ is an optimization problem, then this also leads us to the definition of the \emph{standard parameterization} $\Pi[k]$ of $\Pi$, which is a decision problem containing exactly those tuples $(x, k) \in \Sigma^* \times \N$ for which it holds that $\OPT_\Pi(x') \leq k$ for $\Pi$ being a minimization problem, respectively $\OPT_\Pi(x') \geq k$ if $\Pi$ is a maximization problem.
    
    A \emph{kernelization} for a parameterized problem $Q$ is a polynomial-time algorithm $\mathcal{A}$ which is given as input an instance $(x, \ell)$ and outputs an \emph{equivalent} instance $(x', \ell')$, i.e., $(x, \ell) \in Q \Leftrightarrow (x', \ell') \in Q$, such that $|x'|$ and $\ell'$ are upper bounded by $f(\ell)$ for some computable function $f$. The function $f$ is called the \emph{size} of the kernelization $\mathcal{A}$, and if $f$ is polynomially bounded, then $\mathcal{A}$ is called a \emph{polynomial kernelization}. If there is a probability that the output instance $(x', \ell')$ is not equivalent to $(x, \ell)$, then $\mathcal{A}$ is called a \emph{randomized kernelization}. In our cases, the error probability will be upper bounded by $\Oh(2^{-\Theta(|x|)})$.
    The output instance $(x', \ell')$ is called a \emph{kernel} of $(x, \ell)$ and we say that the problem $Q$ \emph{admits a (randomized) (polynomial) kernel}, if there exists a (randomized) (polynomial) kernelization for $Q$.
    
    \subparagraph{Graphs and graph problems.}
    We use standard graph theoretic notation mostly following Diestel \cite{Diestel2025-GraphTheory}. Our graphs might be directed or undirected, but are always finite, simple, loopless, and unweighted, unless explicitly stated otherwise. A \emph{mixed} graph contains both directed and undirected edges. A \emph{(vertex) annotated graph} is a tuple $(G, \mathcal{T})$, where $G$ is an undirected graph and $\mathcal{T}$ itself is a tuple of subsets of $V(G)$. A \emph{decision or optimization problem on (annotated) graphs} consists of instances encoding (annotated) graphs, respectively. A problem on graphs, i.e., without annotation $\mathcal{T}$, is also called a \emph{pure graph problem}.
    
    An edge $e = \{u, v\} \in E(G)$, respectively $e = (u, v)$, is said to be \emph{incident with} $u$ and $v$, and $e$ is an edge \emph{between} $u$ and $v$. The vertices $u$ and $v$ are called the \emph{endpoints} of $e$, as well as \emph{neighbors of} and \emph{adjacent to} each other. The neighborhood of a vertex $v \in V(G)$ is the set of all neighbors of $v$ in $G$ and is denoted by $N_G(v)$ and $N(v)$ if $G$ is clear from context. Additionally, for a vertex set $V' \subseteq V(G)$, we denote by $N_{G, V'}(v)$ the set of neighbors of $v$ that are contained in $V'$, i.e., $N_{G, V'}(v) := N_G(v) \cap V'$. For vertex set $W \subseteq V(G)$, we extend the definition of the neighborhood of $W$ as $N_G(W) := \bigcup_{v \in W} N_G(v) \setminus W$. The sets $N(W)$ and $N_{G, V'}(W)$ are defined in the same manner. For a set $W \subseteq V(G)$ we define the graph $G[W]$ as the \emph{induced subgraph} of $G$ consisting of the vertex set $W$ and every edge of $G$ between vertices in $W$. Reversely, the graph $G-W$ is equal to the induced subgraph $G[V(G) \setminus W]$ of $G$. A set $S \subseteq V(G)$ is called a \textit{vertex cover} for graph $G$ if $S$ intersects each edge of $G$, or equivalently, if $G-S$ contains no edges. A vertex set $I \subseteq V(G)$ for which $G[I]$ contains no edges is called an \emph{independent set}. 
    In the minimization problem \vertexcover the value of a solution $s$ for undirected graph $G$ is equal to $|S|$ if $s$ encodes a vertex cover $S$ of $G$, and $+\infty$ otherwise.
    
    A \textit{path} $P$ of length $k$ in graph $G$ is a sequence of distinct vertices $(v_0, \dots, v_k)$ such that in $G$ there exist the edges $\{v_0, v_1\}, \{v_1, v_2\}, \dots, \{v_{k-1}, v_k\}$ if $G$ is undirected, respectively edges $(v_0, v_1), (v_1, v_2), \dots, (v_{k-1}, v_k)$ if $G$ is directed. The vertices $v_0$ and $v_k$ are called the \textit{endpoints} of $P$ and the vertices $v_1 \dots, v_{k-1}$ its internal vertices. We also say that the path $P$ lies \textit{between} $v_0$ and $v_k$. With $k$ being odd or even, we say that $P$ has \textit{odd} or \textit{even length}, respectively. We remark that an odd length path consists of an even number of vertices. A set $S \subseteq V(G)$ is called a \textit{(multiway) cut} for a graph $G$ and a set $T = \{t_1, \dots, t_k\} \subseteq V(G)$ of given \textit{terminals} if in graph $G' = G-S$ for every pair $u, v \in T$ there is no path between $u$ and $v$ in $G-S$. More generally, if we are given graph $G$ and a tuple of terminal sets $\mathcal{T} = (T_1, \dots, T_k)$, then a cut for $(G, \mathcal{T})$ is a set $S \subseteq V(G)$ such that for any $i \neq j \in [k]$ and any $u \in T_i, v \in T_j$ there is no path between $u$ and $v$ in $G-S$. We remark that unless stated otherwise, a cut $S$ might intersect the set $T$ of terminals, respectively the set $\bigcup_{T \in \mathcal{T}}T$ in the general case.

    A set $S \subseteq V(G)$ is said to be \emph{closest} to some set $T \subseteq V$ in $G$ if $S$ is the unique $(T, S)$-min cut in $G$, and hence, there exist $|S|$-many vertex-disjoint paths from $T$ to $S$. A \emph{$T$-closest cut} between $T$ and $S$ is a $(T, S)$-min cut that is closest to $T$. Kratsch and Wahlstr\"om \cite{KratschWahlstrom2020-RepresentativeSetsIrrelevantVertices} describe a simple efficient algorithm for computing a $T$-closest cut.
    In the minimization problem \smultiwaycut the value of a solution $x$ for undirected graph $G$ and set $T = \{t_1, \dots, t_s\}$ of terminals, is equal to $|S|$ if $x$ encodes a cut $S$ for $(G, T)$, which is furthermore disjoint from $T$, and $+\infty$ otherwise. In contrast to this, in the minimization problem \dtmultiwaycut we are given an arbitrarily large set $T$ of terminals, but now $x$ is allowed to encode a cut $S$ for $(G, T)$ which might intersect the terminal set $T$, in which case the value of $x$ is also equal to $|S|$.
    
    A \textit{cycle} $C$ of length $k$ in graph $G$ is a sequence of vertices $(v_1 \dots, v_k)$ such that $(v_1, \dots, v_k)$ is a path in $G$ and additionally $G$ contains the edge $\{v_k, v_1\}$, respectively $(v_k, v_1)$ if $G$ is directed. As with paths, the \textit{length} of a cycle is defined by the number of edges and thus $C$ is an \textit{odd cycle} if $k$ is odd and otherwise $C$ is an \textit{even cycle}. Note that we started counting the vertices of a path by zero, while we started counting with one for the cycle. Also observe that for $C$ being an odd cycle, for any $i \leq j \in [k]$ there exist the paths $(v_i, v_{i+1}, \dots, v_{j-1}, v_j)$ and $(v_i, v_{i-1}, \dots, v_1, v_k, v_{k-1}, \dots, v_{j+1}, v_j)$ one of which has odd length and the other has even length. A set $S \subseteq V(G)$ is called an \textit{odd cycle transversal} for a graph $G$ if in graph $G' = G-S$ there exists no odd cycle.  A graph $G$ without any odd cycles is called \emph{bipartite} because its vertex set can be partitioned into the sets $L$ and $R$ such that $G[L]$ ad $G[R]$ are independent sets, i.e., all edges of $G$ have one endpoint in $L$ and the other in $R$. 
    In the minimization problem \oddcycletransversal the value of a solution $s$ for undirected graph $G$ is equal to $|S|$ if $s$ encodes an odd cycle transversal $S$ for $G$, and $+\infty$ otherwise.
    
    A \emph{maximum matching} for a graph $G$ is a set $M \subseteq E(G)$ of edges in $G$ such that no two edges of $M$ share any endpoint. If for some disjoint vertex sets $U, W \subseteq V(G)$ every edge $e \in M$ is incident with one vertex in $U$ and one in $W$, then $M$ is said to be a matching \emph{between} $U$ and $W$. Furthermore, if for every vertex $w \in W$ the matching $M$ contains an edge that is incident with $w$, we say that $W$ is \emph{saturated} by $M$. A crown in a graph $G$ is a pair of disjoint vertex sets $(I, H)$ where (i) $I$ is an independent set of $G$; (ii) $H = N(I)$; (iii) there exists a matching $M$ between $I$ and $H$ which saturates $H$. The problem \problem{Maximum Matching} is a maximization problem on undirected graphs. Here the value of a solution $s$ is equal to $|M|$ if $s$ encodes a matching $M$ for $G$, and $-\infty$ otherwise. It is known under the name \emph{K\H{o}nig's Theorem} that in a bipartite graph, the optimum values for \vertexcover and \problem{Maximum Matching} are equal.
    
    \subparagraph{Boundaried graphs and boundaried kernelization.}
    A \emph{boundaried graph} is a special kind of annotated graph $G_B = (G, B)$, where we call $G$ the \emph{underlying graph} and $B$ the \emph{boundary}. We define the problems \vertexcover, \smultiwaycut, \dtmultiwaycut, \oddcycletransversal, and \problem{Maximum Matching} on boundaried graphs in the same way as on their underlying graphs, i.e., we ignore the boundary for these problems. Especially in situations when we are only interested in the underlying graphs, we will refer to boundaried graphs simply as graphs. For two boundaried graphs $G_B, H_C$ we define the operation of \emph{gluing} these graphs together, denoted by $G_B \oplus H_C$, which results in a new boundaried graph with boundary $B \cup C$ and consisting of the disjoint union of the vertex and edges sets of $G$ and $H$, under the identification of the vertices in $B \cap C$. Throughout this paper we will tacitly assume that $V(G)$ and $V(H)$ intersect exactly at $B \cap C$, in which case the underlying graph of $G_B \oplus H_C$ can simply be seen as the union of each the vertex sets and the edge sets of $G$ and $H$. Clearly, graph gluing is commutative and associative.
    
    Based on regular graph gluing, we additionally define gluing of boundaried graphs with additional vertex annotations, which we will call \emph{(vertex) annotated boundaried graphs}: For two such annotated boundaried graphs $(G_B, T_1, \dots, T_r)$ and $(H_C, U_1, \dots, U_r)$ for some $r \in \N$ and with $V(G) \cap V(H) \subseteq B \cap C$, the result of $(G_B, T_1, \dots, T_r) \oplus (H_C, U_1, \dots, U_r)$ is the annotated boundaried graph $(G_B \oplus H_C, W_1, \dots, W_r)$ where $W_i = T_i \cup U_i$ for every $i \in [r]$.
    
    We say that the instances $(G_B, \mathcal{T})$ and $(G'_B, \mathcal{T}')$ are \emph{gluing equivalent} with respect to an optimization problem $\Pi$ on annotated graphs and boundary $B$, denoted as $(G_B, \mathcal{T}) \equiv_{\Pi, B} (G'_B, \mathcal{T}')$, if there exists some constant $\Delta \in \Z$ such that for every boundaried graph $H_B$ and tuple $\mathcal{U}$ of $V(H)$-subsets, it holds that $\OPT_\Pi((G_B, \mathcal{T}) \oplus (H_B, \mathcal{U})) = \OPT_\Pi((G'_B, \mathcal{T}') \oplus (H_B, \mathcal{U})) + \Delta$. If $\Pi$ and $B$ are clear from context, we will omit them and write $\equiv$ instead. The optimization problem $\Pi$ has \emph{finite integer index} if the number of equivalence classes of $\equiv_{\Pi, B}$ is upper bounded by some function $f(|B|)$.
    
    For the sake of completeness, let us also define the gluing equivalence and finite index for $\Pi$ being a decision problem, although we will only work with optimization problems later on: With $\Pi$ being a decision problem, the instances $(G_B, \mathcal{T})$ and $(G'_B, \mathcal{T}')$ are said to be gluing equivalent with respect to $\Pi$ and $B$, if for every boundaried graph $H_B$ and tuple $\mathcal{U}$ of $V(H)$-subsets it holds that $(G_B, \mathcal{T}) \oplus (H_B, \mathcal{U}) \in \Pi$ if and only if $(G'_B, \mathcal{T}') \oplus (H_B, \mathcal{U}) \in \Pi$. If the number of equivalence classes of $\equiv_{\Pi, B}$ is upper bounded by some function $f(|B|)$, then the decision problem $\Pi$ is said to have \emph{finite index}.
    
    Let $\Pi$ be an optimization or decision problem and $\rho$ a minimization problem, both on annotated graphs. A \emph{boundaried kernelization} for the parameterized problem $\Pi[\rho]$ is a polynomial-time algorithm $\mathcal{A}$ which is given as input an annotated boundaried graph $(G_B, \mathcal{T})$ and feasible solution $s$ for $\rho$ on $(G_B, \mathcal{T})$ together with the value $\ell$ of $s$; and $\mathcal{A}$ outputs an annotated boundaried graph $(G'_B, \mathcal{T}')$, such that $(G'_B, \mathcal{T}')$ is gluing equivalent to $(G_B, \mathcal{T})$ and the size of $(G'_B, \mathcal{T}')$ is bounded by $f(|B| + \ell)$ for some computable function $f$. Additionally, if $\Pi$ is an optimization problem, then $\mathcal{A}$ also needs to output the offset $\Delta$ by which the optimum value for $(G_B, \mathcal{T}) \oplus (H_B, \mathcal{U})$ differs from that for $(G_B', \mathcal{T}') \oplus (H_B, \mathcal{U})$ for any boundaried graph $H_B$ and tuple $\mathcal{U}$ of $V(H)$-subsets. Analogously to (regular) kernelization, we define the size of $\mathcal{A}$ and the notions of a (randomized) (polynomial) boundaried kernelization and kernel. By minor adaptation of a result by Antipov and Kratsch \cite[Lemma 4]{AntipovKratsch2025-BoundariedKernelization} and under the assumption that $\Pi$ and $\rho$ are \NP-optimization problems, the existence of a (randomized) (polynomial) boundaried kernelization for parameterized problem $\Pi[\rho]$ also implies the existence of a (randomized) (polynomial) kernelization for $\Pi[\rho]$. Additionally we define a boundaried kernelization for $\Pi[\text{local solution}]$ to be a boundaried kernelization for $\Pi$ parameterized by itself, i.e., the input is an annotated boundaried graph $(G_B, \mathcal{T})$ and a feasible solution $s$ for $\Pi$ on $(G_B, \mathcal{T})$, together with the value $\ell$ of $s$.
    
    We will construct boundaried kernelization by the use of reduction rules. We say that a reduction rule is \emph{gluing safe}, if it gets as input an instance $(G_B, \mathcal{T})$ and outputs a gluing equivalent instance $(G'_B, \mathcal{T}')$ and the corresponding shift $\Delta$ in solution value, if $\Pi$ is an optimization problem and $\Delta \neq 0$.
    Furthermore, we will make use of the following two lemmas.
    
    \begin{lemma}[{\cite[Lemma 1]{AntipovKratsch2025-BoundariedKernelization}}]\label{lem:bk_superset}
        Let $\Pi$ be a pure graph problem, $G$ and $G'$ graphs, and $B, C$ vertex subsets of both $V(G)$ and $V(G')$ such that $B \subseteq C$. Then $G_C \equiv_{\Pi, C} G'_C$ implies $G_B \equiv_{\Pi, B} G'_B$, with the same offset $\Delta$ for these two equivalences, if $\Pi$ is an optimization problem.
    \end{lemma}
    
    \begin{lemma}[{\cite[Lemma 5]{AntipovKratsch2025-BoundariedKernelization}}]\label{lem:bk_lb}
        Let $\Pi$ be a problem on (annotated) graphs. Let $f$ be a computable function such that \parameterizedproblem{$\Pi$}{local solution} admits a boundaried kernelization of size $f$. Further fix some vertex set $B$ and family $\mathcal{C}$ of boundaried $\Pi$-instances such that for each $(G_B, \mathcal{T}) \in \mathcal{C}$ there exists some $\Pi$-solution $s$ with value at most $g(|B|)$ for some function $g$. Then the number of equivalence classes of $\equiv^\mathcal{C}_{\Pi, B}$ is upper bounded by $\Oh(2^{f(|B| + g(|B|))^2})$.
    \end{lemma}
    
    \section{Matroid tools for kernelization}\label{section:matroid}
    In this section we recall some of the results of Kratsch and Wahlstr\"om~\cite{KratschWahlstrom2020-RepresentativeSetsIrrelevantVertices} regarding the use of matroids for polynomial kernelization, as well as the corresponding definitions.
    
    A \emph{matroid} is a pair $M = (E, \cI)$ of a \emph{ground set} $E$ and a collection of \emph{independent sets} $\cI \subseteq 2^E$, which are subsets of $E$, such that: (i) $\emptyset \in \cI$; (ii) if $I_1 \subseteq I_2$ and $I_2 \in \cI$, then also $I_1 \in \cI$; and (iii) if $I_1, I_2 \in \cI$ and $|I_2| > |I_1|$, then there exists some $x \in I_2 \setminus I_1$ such that $I_1 \cup \{x\} \in \cI$. An independent set $B \in \cI$ is called a \emph{basis} of $M$ if no superset of $B$ is independent. One can also define matroid $M$ by its set of bases. The \emph{rank} $r(X)$ of a subset $X \subseteq E$ is the largest cardinality of an independent set $I \subseteq X$. The rank of $M$ is $r(M) := r(E)$.
    Some matroid $M$ is said to be \emph{linear}, if there exists some $m \times n$ matrix over a field $\mathbb{F}$, such that $M = (E, \mathcal{I})$, where $E$ is the set of columns of $A$, and $\mathcal{I}$ contains those subsets of $E$ that are linearly independent over $\mathbb{F}$. In such a case we also say that $A$ \emph{represents} $M$, and $A$ \emph{is a representation of} $M$.
    
    Let $D = (V, A)$ be a directed graph and let $S, T \subseteq V$. We say that $T$ is \emph{linked} to $S$ in $D$ if there exist $|T|$-many vertex-disjoint paths from $S$ to $T$, also allowing paths of length zero. In particular, any set is linked to itself. Given a directed graph $D = (V, A)$, a set $S \subseteq V$ of source vertices, and a set $U \subseteq V$ of sink vertices, we define a special case of matroids, called \emph{gammoid}, by the sets $T \subseteq U$ that are linked to $S$ in $D$. We will only use a further special case thereof with $U = V$, called \emph{strict gammoid}, and slightly abuse notation by simply calling these gammoids. Due to Perfect~\cite{Perfect1968-ApplicationsMengersGraph} and Marx~\cite{Marx2009-ParameterizedViewMatroid}, a representation of a (strict) gammoid can be computed in randomized polynomial time, with one-sided error that can be made exponentially small in the size of the graph.
    
    Let $M = (V, \cI)$ be a matroid and let $X$ be an independent set in $M$. A set $Y \subseteq V$ is said to \emph{extend} $X$ in $M$ if it holds that $X \cap Y = \emptyset$ and $X \cup Y \in \cI$. For a collection of $V$-subsets $\mathcal{Y} \subseteq 2^V$ we say that a subset $\mathcal{Y}' \subseteq \mathcal{Y}$ is $r$-representative for $\mathcal{Y}$ if for every set $X \subseteq V$ of size at most $r$, the existence of a set $Y \in \mathcal{Y}$ that extends $X$ in $M$, implies the existence of a set $Y' \in \mathcal{Y}'$ that also extends $X$ in $M$.
    
    Next, we state a result by Marx \cite{Marx2009-ParameterizedViewMatroid} dubbed as the \emph{representative sets lemma} by Kratsch and Wahlstr\"om, and follow that by results that build up on it.
    
    \begin{lemma}[\cite{Marx2009-ParameterizedViewMatroid,KratschWahlstrom2020-RepresentativeSetsIrrelevantVertices}]
        Let $M$ be a linear matroid of rank $r+s$, and let $\mathcal{Y}$ be a collection of independent sets of $M$, each of size $s$. There exists a set $\mathcal{Y}'\subseteq \mathcal{Y}$ of size at most $\binom{r+s}{s}$ that is $r$-representative for $\mathcal{Y}$. Furthermore, given a representation $A$ of $M$, we can find such a set $\mathcal{Y}'$ in time $(m+||A||)^{\Oh(s)}$, where $m = |\mathcal{Y}|$ and $||A||$ denotes the encoding size of $A$.
    \end{lemma} 
    
    \begin{lemma}[{\cite[Theorem 5]{Kratsch2014-RecentDevelopmentsKernelization}}]\label{lem:repset}
        Let $M$ be a gammoid and let $\mathcal{A} = \{A_1, \dots, A_m\}$ be a collection of independent sets, each of size $p$. We can find in randomized polynomial time a set $\mathcal{A}' \subseteq \mathcal{A}$ of size at most $\binom{p + q}{p}$ that is $q$-representative for $\mathcal{A}$.
    \end{lemma}
    
    \begin{lemma}[{\cite[Theorem 1.2]{KratschWahlstrom2020-RepresentativeSetsIrrelevantVertices}}]\label{lem:cutset_small} 
        Let $G = (V, E)$ be a directed graph and let $S, T \subseteq V$. Let $r$ denote the size of a minimum $(S, T)$-vertex cut (which may intersect $S$ and $T$). There exists a set $Z \subseteq V$ with $|Z| = \Oh(|S| \cdot |T| \cdot r)$, such that for any $A \subseteq S$ and $B \subseteq T$, it holds that $Z$ contains a minimum $(A, B)$-vertex cut as a subset. We can find such a set in randomized polynomial time with failure probability $\Oh(2^{-|V|})$.
    \end{lemma}
    
    \begin{lemma}[{\cite[Corollary 1.5]{KratschWahlstrom2020-RepresentativeSetsIrrelevantVertices}}]\label{lem:cutset}
        Let $G = (V, E)$ be an undirected graph and $X \subseteq V$. Let $s$ be a constant. Then there is a set $Z \subseteq V$ of $\mathcal{O}(|X|^{s+1})$ vertices such that for every partition $\mathcal{X} = (X_0, X_X, X_1, \dots, X_s)$ of $X$ into $s+2$ possibly empty parts, the set $Z$ contains a minimum multiway cut of $(X_1, \dots, X_s)$ in the graph $G - X_X$ as a subset. We can compute such a set in randomized polynomial time with failure probability $\mathcal{O}(2^{-|V|})$.
    \end{lemma}
    
    Such a set $Z$, as described in the preceding two lemmas, is called a \emph{cutset} for the pair of vertex sets $(S, T)$ (Lemma~\ref{lem:cutset_small}), respectively for the vertex set $X$ (Lemma~\ref{lem:cutset}).
    
    The main operation for reducing our graphs will be that of \emph{bypassing} a vertex $v$ in a graph $G$ (also called making vertex $v$ \emph{undeletable} in $G$). In this operation one removes the vertex $v$ from $G$ and adds shortcut edges between the neighbors of $v$. In other words, for all paths $(u, v, w)$ in $G$, one adds the edge $(u, w)$ or $\{u, w\}$ (depending on whether the graph is (un)directed or mixed) unless already present. Effectively, this operation forbids to take $v$ into a cut, while preserving the separation achieved by any vertex cuts that avoid $v$. For more detail, see \cite{KratschWahlstrom2020-RepresentativeSetsIrrelevantVertices}.
    Under the term \emph{bypassing a set $W \subseteq V(G)$}, we mean the repeated operation of bypassing vertices of $W$ one after another, in any order. Observe that by repeated use, the following lemma also holds when bypassing a vertex set $W$, if $X \subseteq V \setminus W$.
    
    \begin{lemma}[{\cite[Proposition 2.3]{KratschWahlstrom2020-RepresentativeSetsIrrelevantVertices}}]\label{lem:bypass}
        Let $G = (V, E)$ be an undirected, directed, or mixed graph and let $G'$ be the result of bypassing some vertex $v \in V$ in $G$. Then, for any set $X \subseteq V \setminus \{v\}$ and any vertices $s, t \in V \setminus (X \cup \{v\})$, there is a path from $s$ to $t$ in $G - X$ if and only if there is a path from $s$ to $t$ in $G' - X$.
    \end{lemma}
    
    \section{$s$-Multiway Cut[local solution]}\label{section:smwc}
    As our first result, we construct a boundaried kernelization variant of the randomized polynomial kernel for the \smultiwaycut problem by Kratsch and Wahlstr\"om \cite{KratschWahlstrom2020-RepresentativeSetsIrrelevantVertices}. First, let us take a look at the main idea of their kernelization algorithm. Let us be given a graph $G$ and a terminal set $T$. We can assume that no two terminals are adjacent and that $|T| \leq s$, else there exists no feasible solution and we can output a constant-size instance without any feasible solution. Now, let us additionally fix any feasible solution $Y$. It is apparent that for each terminal $t \in T$ every neighbor $x \in N(t)$ is contained either in $Y$ or in the same component as $t$ in $G - Y$. That is where Kratsch and Wahlstr\"om make use of their cut cover result Lemma~\ref{lem:cutset} in order to find a set $Z$ of size $\Oh(|N(T)|^{s+1})$ such that for any combination of forcing vertices from $N(T)$ into the solution or into the components of individual terminals, there exists an optimum solution corresponding to that choice, which is in addition totally contained in $Z$. This leads to the correctness of bypassing any vertex which is not contained in $Z \cup T \cup N(T)$. By an additional reduction through an LP-based approach of Guillemot \cite{Guillemot2011-FPTAlgorithmsPathtransversal}, Kratsch and Wahlstr\"om bound the size of $N(T)$ by $\Oh(k)$ beforehand, with $k$ being the sought solution size. This results in a total bound of $\Oh(k^{s+1})$ for the vertex size of the kernel.
    
    In the case of boundaried kernelization, we are only given a local part $G_B$ of the graph and only a subset $T_G$ of the terminals, while the rest $H_B$ of the graph and $T_H$ of the terminal set are not known to the kernelization algorithm. This means that we cannot directly work with every combination of forcing neighbors of terminals into the solution or into the terminals' corresponding components. Actually, this even leads us to Theorem~\ref{thm:s-mwc_BTH_noBK} which rules out any boundaried kernelization for \smultiwaycutlocal in the general case. However, if we assume that any terminals in $B$ must already be contained in $T_G$, i.e., for the gluing equivalence we only consider instances $(H_B, T_H)$ with $T_H \cap B = \emptyset$, then we can prepare for any combination of: (i) forcing neighbors of $T_G$-terminals into the solution or into the corresponding components; and (ii) additionally forcing boundary vertices outside of $T_G$ into the solution or the terminals' components. Again, this step is performed through Lemma~\ref{lem:cutset}, that is, we find a set $Z$ of size $\Oh((|B| + |N(T)|)^{s+1})$ that contains the best local cut for each choice. For bounding the size of $N_G(T_G)$ beforehand, we substitute the LP-based approach by Reduction Rule~\ref{rr:smwc_neighbor}, and get a total size of $\Oh((|B| + \ell)^{s+1})$ vertices for the randomized boundaried kernelization.
    
    \begin{theorem}\label{thm:s-mwc_BTH_noBK}
        In general, \smultiwaycutlocal does not admit a boundaried kernelization,
    \end{theorem}
    \begin{proof}
        Let $B = \{x, y\}$. We construct two instances $(G^i_B,\emptyset), (G^j_B,\emptyset)$ for each distinct pair $i, j \in \mathbb{N}$ such that it holds that $(G^i_B,\emptyset) \not\equiv (G^j_B,\emptyset)$ and both have local solution size zero. By Lemma~\ref{lem:bk_lb} this rules out a boundaried kernelization.
        Let $G^i$ (resp. $G^j$) be a complete bipartite graph $K_{2, i}$ (resp. $K_{2, j}$) with parts $\{x, y\} = B$ and $\{v_1, \dots, v_i\}$ (resp. $\{v_1, \dots, v_j\}$). In order to exclude gluing equivalence of these annotated graphs, we define the graph $H$ which only consists of the independent set $B = \{x, y\}$. 
        Now, observe that it holds that $\OPT((G^i_B, \emptyset) \oplus (H_B, \emptyset)) = 0 = \OPT((G^j_B, \emptyset) \oplus (H_B, \emptyset))$, while $\OPT((G^i_B, \emptyset) \oplus (H_B, \{x, y\}) = i \neq j = \OPT((G^j_B, \emptyset) \oplus (H_B, \{x, y\})$. Hence, for any $i \neq j$ it holds that $(G^i_B, T_G) \not\equiv (G^j_B, T_G)$.
    \end{proof}
    
    \begin{theorem}\label{thm:smultiwaycut}
        The parameterized problem \smultiwaycutlocal admits a randomized polynomial boundaried kernelization with  $\Oh((|B| + \ell)^{s+1})$ vertices, with failure probability $\Oh(2^{-|V|})$, and under the restriction that externally glued instances do not contain  $B$-vertices as terminals.
    \end{theorem}
    
    As partly indicated above, the input to the boundaried kernelization is a tuple $(G_B, T_G, S, \ell)$ such that $|T_G| \leq s$ and $S$ is a multiway cut of size at most $\ell$ for $T_G$ in $G$. In particular it follows that $B' := S \cup (B \setminus T_G)$ is also a multiway cut for $T_G$ in $G$, i.e., every terminal $t \in T_G$ is contained in its own component $C_t$ in $G - B'$.
    
    \begin{rrule}\label{rr:smwc_neighbor}
        For any $t \in T_G$, if $|N(t)| > |B'|$, then choose a minimum size cut $X$ between $N(t)$ and $B'$ (allowing terminal-deletion) which is closest to $N(t)$. Remove the edges incident with $t$, and make $X$ the new neighborhood of $t$.
    \end{rrule}
    \begin{lemma}
        Reduction Rule~\ref{rr:smwc_neighbor} is gluing-safe.
    \end{lemma}
    \begin{proof}
        Let $P$ be the component of $G-X$ which contains $t$, and thus $X = N(P)$. Let $G'$ be the result of the reduction rule. Fix arbitrary $(H, B, T_H)$, and set $T = T_G \cup T_H$. Assume $|T| \leq s$, else there exists no feasible solution for both $(G_B \oplus H_B, T)$ and $(G'_B \oplus H_B, T)$.
        Note that $(G_B \oplus H_B) - P = (G'_B \oplus H_B) - P$, since we only changed edges that are incident with $t$, which is contained in $P$.
        We prove gluing equivalence of $(G, T_G)$ and $(G', T_G)$ by showing that $\OPT(G_B \oplus H_B, T)$ is equal to $\OPT(G'_B \oplus H_B, T)$.
        
        ``$\OPT(G_B \oplus H_B, T) \leq \OPT(G'_B \oplus H_B, T)$'':\\
        Let $Y'$ be any feasible solution for $G'_B \oplus H_B$. We show that $Y'$ is also a solution for $G_B \oplus H_B$.
        For contradiction, assume that there exist two distinct terminals $t_1, t_2 \in T$ with a path between them in $(G_B \oplus H_B) - Y'$. As $(G_B \oplus H_B) - P = (G'_B \oplus H_B) - P$ and $t$ is the only terminal contained in $P$, such a path must go through $P$ and thus through $N_G(P)$. Without loss of generality assume that $t \neq t_1$. Since the assumed path goes through $N_G(P) = N_{G'}(t)$, there is a path from $t_1$ to $t$ in $(G'_B \oplus H_B) - Y'$, a contradiction.
        
        ``$\OPT(G_B \oplus H_B, T) \geq \OPT(G'_B \oplus H_B, T)$'':\\
        Let $Y$ be a feasible solution for $G_B \oplus H_B$ and let $Q$ be the component of $(G_B \oplus H_B) - Y$ that contains $t$. Note that $t$ is the only terminal both in $P$ and in $Q$, i.e., it holds that $T \cap (P \cup Q) = T \cap (P \cap Q) = \{t\}$.
        We define the set $Y' := (Y \setminus N_G(Q)) \cup N_G(P \cup Q) = (Y \setminus (N_G(Q) \cap P)) \cup (N_G(P) \setminus Q)$.
        
        First, note that $|Y'|$ is upper bounded by $|Y|$: By submodularity of $|N_G(\cdot)|$, it holds that  $|N_G(P)| + |N_G(Q)| \geq |N_G(P \cap Q)| + |N_G(P \cup Q)|$. As $t \in (P \cap Q)$, it holds that $N_G(P \cap Q)$ is a cut between $N_G(t)$ and $B'$. Since $X = N_G(P)$ is closest to $N_G(t)$, it follows that $|N_G(P)| \leq |N_G(P \cap Q)|$. Putting this into the earlier equation, we get $|N_G(Q)| \geq |N_G(P \cup Q)|$, which results in $|Y'| \leq |Y| - |N_G(Q)| + |N_G(P \cup Q)| \leq |Y|$.
        
        Next, we make sure that $Y'$ is a $s$-multiway cut for $(G'_B \oplus H_B, T)$: 
        First we show that there is no path between some $t_1 \in T \setminus \{t\}$ and $t$ in $(G'_B \oplus H_B) - Y'$: Assume the contrary is true. Clearly, such a path needs to pass some vertex $x \in N_{G'_B \oplus H_B}(t) = N_G(P) \cup N_H(t)$ (with $N_H(t) = \emptyset$ if $t \in V(G)\setminus B$). Without loss of generality, choose $x$ such that on the subpath from $t_1$ to $x$, there are no other vertices from $N_{G'_B \oplus H_B}(t)$. Since by definition of $Y'$ it holds that $N_G(P) \subseteq Q \cup Y'$ and $N_H(t) \subseteq (Q \cup Y) \cap V(H) \subseteq Q \cup Y'$, it follows that $x \in N_{G'_B \oplus H_B}(t) \setminus Y' \subseteq Q$. By choice of $P$, the terminal $t_1$ is not contained in $P$, and neither are the vertices of the subpath between $t_1$ and $x$ contained in $P \cup N_{G'}(P) = P \cup N_G(P)$. From $Y \bigtriangleup Y' \subseteq P \cup N_G(P)$ it follows that $P \cup N(P) \cup Y = P \cup N(P) \cup Y'$ and thus $(G_B \oplus H_B) - (P \cup N(P) \cup Y) = (G'_B \oplus H_B) - (P \cup N(P) \cup Y')$. But this means that the path from $t_1$ to $x$ also exists in $(G_B \oplus H_B) - Y$, which contradicts the fact that $Y$ is a cut between $t_1$ and $Q$ in $G_B \oplus H_B$.
        In a similar manner we can be sure that no path exists between any distinct pair of terminals $t_1, t_2 \in T \setminus \{t\}$ in $(G'_B \oplus H_B) - Y'$: As $t_1, t_2$ are both not contained in $P$ and it holds that $(G_B \oplus H_B) - (P \cup N(P) \cup Y) = (G'_B \oplus H_B) - (P \cup N(P) \cup Y')$, such a path would need to pass some vertex $x \in N(P)$. But we have seen above already, that a path from $t_1$ to $x$ in $(G'_B \oplus H_B) - Y'$ cannot exist.
    \end{proof}
    
    \begin{rrule}\label{rr:smwc_cutcover}
        Assume that Reduction Rule~\ref{rr:smwc_neighbor} cannot be applied.
        Let $Z \subseteq V(G)$ be a vertex set such that for every partition $\mathcal{X} = (X_0, X_X, X_1, \dots, X_s)$ of $X := B \cup \bigcup_{t \in T_G} N_G(t)$ into $s+2$ possibly empty parts, the set $Z$ contains a minimum multiway cut of $(X_1, \dots, X_s)$ in the graph $G - X_X$. Bypass every vertex from $V(G) \setminus (Z \cup B \cup T_G \cup N_G(T_G))$.
    \end{rrule}
    
    \begin{lemma}
        Reduction Rule~\ref{rr:smwc_cutcover} is gluing-safe.
    \end{lemma}
    \begin{proof}
        Let $G_B$ denote the graph before applying the reduction rule, and $G'_B$ the result thereof. Let $H_B$ be the graph for gluing. Let $T_G$ be the given set of terminals for $G$ (resp.\ for $G'$), and $T_H \subseteq V(H) \setminus B$ the set of terminals given together with $H_B$. Let $T = T_G \cup T_H$ and assume without loss of generality that $|T| = s$. If $|T| < s$, we can add isolated terminals to $H$, while $|T| > s$ directly leads to $(G_B \oplus H_B, T)$ having no feasible solution. Likewise, we can assume that no two terminals are adjacent.
        Let $s_G = |T_G|$ and assume some arbitrary ordering of the terminals, such that the vertices of $T_G$ can be denoted by $t_1, \dots, t_{s_G}$ and those from $T_H$ by $t_{s_G + 1}, \dots, t_s$.
        
        ``$\OPT(G_B \oplus H_B, T) \leq \OPT(G'_B \oplus H_B, T)$'':
        Let $Y'$ be a multiway cut for $T$ in $G'_B \oplus H_B$. Since $Y'$ does not contain any bypassed vertices, it follows from Lemma~\ref{lem:bypass} that $Y'$ is a multiway cut for $T$ in $G_B \oplus H_B$.
        
        ``$\OPT(G_B \oplus H_B, T) \geq \OPT(G'_B \oplus H_B, T)$'':
        Let $Y$ be a multiway cut for $T$ in $G_B \oplus H_B$ without terminal deletion. Thus, no component in $(G_B \oplus H_B) - Y$ contains more than one terminal and, obviously, every terminal $t \in T_G$ is in the same component as $N_G(t) \setminus Y$. Furthermore, any vertex in $B \setminus Y$ shares its component with at most one terminal. Based on this we partition the set $X$ into $(X_0, X_X, X_1, \dots, X_s)$ as follows: $X_X$ contains those $B \setminus T$-vertices that are contained in $Y$ and $X_0$ those that are not connected with any terminal in $(G_B \oplus H_B) - Y$; for each $t_i \in T$ put into $X_i$ every $B \setminus Y$-vertex that is connected with $t_i$ in $(G_B \oplus H_B) - Y$, and if further $t_i$ is contained in $T_G$ then also put all vertices from $N_G(t_i) \setminus Y$ into $X_i$.
        
        Let $F = V(G) \setminus B$. Clearly, $Y \cap F$ is a multiway cut for $(X_1, \dots, X_s)$ in $G - X_X$. By choice of $Z$, construction of $G'$ and Lemma~\ref{lem:bypass}, there is also a multiway cut $Y^* \subseteq Z$ of size at most $|Y \cap F|$ for $(X_1, \dots, X_s)$ in $G' - X_X$. We will now verify that $Y' := (Y \setminus F) \cup Y^*$ is a multiway cut for $T$ in $G'_B \oplus H_B$ (without terminal deletion). For that, assume for contradiction that there is a path $P$ between distinct terminals $t_{i_1}, t_{i_2} \in T$ in $(G'_B \oplus H_B) - Y'$. Note that no two terminals from $T$ are connected in either $H - Y' = H - Y$ or $G' - Y'$. Thus, the path cannot go only through $H-B$ or $G'-B$ and contains vertices from $B \setminus (T \cup Y')$. Give these vertices an order along the path in the direction from $t_{i_1}$ to $t_{i_2}$, i.e., $P = (t_{i_1}, \dots, x_1, \dots, x_2, \dots, x_p, \dots, t_{i_2})$, where $x_1, \dots, x_p$ are the $B \setminus (T \cup Y')$-vertices contained in $P$ and the vertices between $x_i, x_j$ with distinct $i, j \in [p]$ are contained in $(V(G) \cup V(H)) \setminus (B \cup T \cup Y')$. If the path between $x_i$ and $x_j$ goes through $H - (B \cup Y')$, then they need to be in the same part of the earlier partition $(X_0, X_X, X_1, \dots, X_s)$: Since it holds that $H -(B \cup Y') = H - (B \cup Y)$, the vertex $x_i$ is connected to $t \in T$ in $(G_B \oplus H_B) - Y$ if and only if $x_j$ is connected to $t$. However, if the path goes through $G' - (B \cup Y')$ and if, say, $x_i$ is in the part $X_h$ with $t_h \in T_H$, then it might also be the case that $x_j$ is in the part $X_0$, since $Y^*$ does not have to separate $X_h$-vertices from $X_0$-vertices. Furthermore, if terminal $t_h \in T$ is connected with $x_i$ (either through $H - (B \cup Y')$ or $G' -(B \cup Y')$), then $x_i$ needs to be contained in the part $X_h$. Combining these points we come to the fact that $x_1 \in X_{i_1}$ and $x_p \in X_{i_2}$, while at the same time it holds that $x_p \in X_{i_1} \cup X_0$ by the transversal from $x_1$ to $x_p$ through $H - (B \cup Y')$ and $G' -(B \cup Y')$ between consecutive $B$-vertices $x_i, x_{i+1}$ for every $i \in [p-1]$. A contradiction.
    \end{proof}
    
    We are now ready to describe the complete algorithm.
    \begin{proof}[Proof of Theorem~\ref{thm:smultiwaycut}]
        Let $G_B$ be the graph, $T_G$ the terminal set and $S$ the local solution of size at most $\ell$ to \smultiwaycut  on $(G, T_G)$. We can assume without loss of generality that $|T_G| \leq s$, since otherwise there exists no feasible solution $S$. Let $F = V(G) \setminus B$ and $B' = (B \cup S) \setminus T_G$. In polynomial time, apply exhaustively Reduction Rule~\ref{rr:smwc_neighbor} to ensure that $|N_G(t)|$ is upper bounded by $|B'| \leq |B| + \ell$. As a next step, apply Lemma~\ref{lem:cutset} in order to compute, with success probability $1-\Oh(2^{-|V|})$, a cut cover $Z$ for $G$ and $X = B \cup \bigcup_{t \in T_G} N_G(t)$ of size at most $\Oh((|B| + \ell)^{s+1})$ for Reduction Rule~\ref{rr:smwc_cutcover}. From the gluing-safeness of the reduction rule and under the assumption of correct computation of cut cover $Z$, we end up with a gluing-equivalent instance $(G'_B, T_G)$ with $\Oh((|B| + \ell)^{s+1})$ vertices.
    \end{proof}
    
    \section{Deletable Terminal Multiway Cut[local solution]}\label{section:dtmwc}
    Quite closely related to the \smultiwaycut problem is \dtmultiwaycut, in which the number of input terminals is arbitrary, but one is also allowed to take them into the solution. 
    Given graph $G$ and terminal $T$, Kratsch and Wahlstr\"om~\cite{KratschWahlstrom2020-RepresentativeSetsIrrelevantVertices} reduce the instance to $\Oh(|T|^3)$ vertices. Using an LP-argument by Guillemot~\cite{Guillemot2011-FPTAlgorithmsPathtransversal} they further reduce the terminal set $T$ to size $\Oh(k)$ with $k$ being the sought solution size. While it is easy to see that this LP-argument does not work for boundaried kernelization on $(G_B, T_G)$, as we have no access to the complete instance, we are still able to apply an argument due to Razgon~\cite{Razgon2011-LargeIsolatingCuts} in order to reduce the size of $T_G$ to $\Oh((|B| + \ell)^2)$, where $\ell$ is the size of a given local \dtmultiwaycut solution $S$ on $(G, T_G)$. The corresponding Reduction Rules~\ref{rr:dtmwc_iso}, \ref{rr:dtmwc_high_deg_F}, and \ref{rr:dtmwc_high_deg_B} remind of removing isolated vertices and those with too high degree for \vertexcover kernelization. Afterwards, we adapt the reduction to $\Oh(|T|^3)$ vertices by Kratsch and Wahlstr\"om and get a bound of $\Oh((|B|+|T_G|)^3) = \Oh((|B|+\ell)^6)$ vertices for the boundaried case.
    
    We give a short overview of the original algorithm and subsequently describe our adaption. Kratsch and Wahlstr\"om show that for any optimal solution $X$ that additionally maximizes the intersection size with terminal set $T$, and for any $x \in X$, there exists a path packing of size $|X|+2$ from $T$ to $X$ and two sink-only copies $x', x''$ of $x$. They compute a representative set $\mathcal{Z}'$ of the set $\mathcal{Z} := \{\{x, x', x''\} \mid x \in V \setminus T\}$ on the gammoid $\mathcal{M}$ defined on $G$ with two sink-only copies for each vertex, with sources $T$, and show, using that earlier connectedness between $T$ and $X$, that for all vertices $x \in X$, it holds that $\{x, x', x''\}$ is contained in $\mathcal{Z}'$. This allows them to bypass any vertex $v \in V \setminus T$ for which $\{v, v', v''\}$ is not contained in $\mathcal{Z}'$, resulting in a reduced instance with $\Oh(k^3)$ vertices.
    
    Now to our boundaried kernelization. Let $(G_B, T_G)$ be our given local instance. Using an argument by Kratsch~\cite{Kratsch2014-RecentDevelopmentsKernelization}, we show that for any graph $H_B$ and terminal set $T_H$, for any optimal solution $X$ for $(G_B \oplus H_B, T_G \cup T_H)$ that additionally maximizes the intersection size with terminal set $T := T_G \cup T_H$, and for any vertex $x \in X_0 \cap F$, there exists a set of $|X_0 \cap F| + 2$ pairwise vertex-disjoint paths from $(B \cup T_G) \setminus X$ to $(X_0 \cap F) \cup \{x', x''\}$ with $F := V(G) \setminus B$, $X_0 = X \setminus T$, and $x', x''$ being sink-only copies of $x$. Using this result, we again show that on a similarly defined gammoid with sources $B \cup T_G$, a representative set $\mathcal{Z}'$ of $\mathcal{Z} = \{\{x, x', x''\} \mid x \in F\}$ with size $\Oh((|B|+|T_G|)^3)$ contains the triplet $\{x, x', x''\} \in \mathcal{Z}'$ for any $x \in X_0 \cap F$. Consequently, we then bypass all vertices not contained in $B \cup T_G \cup \{v \mid \{v, v', v''\} \in \mathcal{Z}'\}$. Together with the earlier bound on $|T_G|$ we thus get the following.
    
    \begin{theorem}\label{thm:pbk_dtmwc}
        The parameterized problem \dtmultiwaycutlocal admits a randomized polynomial boundaried kernel with $\Oh((|B|+\ell)^6)$ vertices.
    \end{theorem}
    
    Let us first reduce the number of terminals by adapting arguments due to Razgon~\cite[Theorem 5]{Razgon2011-LargeIsolatingCuts}.
    
    \begin{rrule}\label{rr:dtmwc_iso}
        Assume $C$ to be a connected component of $G$ with $|C \cap T_G| \leq 1$ and $C \cap B = \emptyset$. Then remove $C$.
    \end{rrule}
    \begin{lemma}
        Reduction Rule~\ref{rr:dtmwc_iso} is gluing safe.
    \end{lemma}
    \begin{proof}
        Note that no matter what graph $H_B$ and terminal set $T_H$ we consider, it holds that $C$ remains a connected components of $G_B \oplus H_B$ and contains at most one terminal from $T_G \cup T_H$. Thus no vertex of $C$ connects at least two terminals, and $C$ is irrelevant for any optimum solution $X$ on $(G_B \oplus H_B, T_G \cup T_H)$.
        
        The other direction is trivial, as any solution for $(G_B \oplus H_B, T_G \cup T_H)$ is also feasible for $((G_B-C) \oplus H_B, T_G \cup T_H)$.
    \end{proof}
    
    \begin{rrule}\label{rr:dtmwc_high_deg_F}
        Assume some vertex $v \in F := V(G) \setminus B$, such that there exist paths from $v$ to $\ell + |B| + 2$ many $T_G$-terminals, and these paths only intersect at $v$. Then remove $v$ and increase $\Delta$ by one.
    \end{rrule}
    \begin{lemma}
        Reduction Rule~\ref{rr:dtmwc_high_deg_F} is gluing safe.
    \end{lemma}
    \begin{proof}
        Fix any graph $H_B$ and terminal set $T_H \subseteq V(H)$. Let $Q = G_B \oplus H_B$ and $T = T_G \cup T_H$. Let $G' = G-v$, $Q' = G'_B \oplus H_B$, $T'_G = T_G \setminus \{v\}$ (which might be equal to $T_G$ in the case that $v \in F \setminus T_G$), and $T' = T'_G \cup T_H$. We show correctness of the equality $\OPT(Q, T) = \OPT(Q', T') + 1$.
        
        The direction $\OPT(Q, T) \leq \OPT(Q', T') + 1$ follows from the simple observation that for any feasible solution $X' \subseteq V(Q')$ it holds that $Q' - X' = Q - (X' \cup \{v\})$ and $T' \setminus X' = T \setminus (X' \cup \{v\})$. Hence, by the fact that there is no path between two terminals of $T' \setminus X'$ in $Q' - X'$, we know that there is no path between two terminals of $T \setminus (X' \cup \{v\})$ in $Q - (X' \cup \{v\})$, implying that $X' \cup \{v\}$ is a feasible solution for $(Q, T)$ of size $|X'| + 1$.
        
        Now we show the direction ``$\OPT(Q, T) \geq \OPT(Q', T') + 1$'':
        Let $X$ be an optimum solution for $(Q, T)$. If $v \in X$, then we are done, hence assume the opposite case. Then $X$ contains one $V(G)$-vertex for each but one of the $\ell + |B| + 2$ many paths from $v$ to $T_G$-terminals, that only intersect at $v$. Note that $\tilde{X} := (X \setminus V(G)) \cup B \cup S$ (with $S$ being the given solution for $(G, T_G)$ of size at most $\ell$) is also a feasible solution for $Q$: any $B$-intersecting paths between terminals of $T$ are hit by $B$, any paths between $T_H$-terminals are hit by $X \setminus V(G)$, and any paths between $T_G$-terminals are hit by $S$. But also it holds that $|\tilde{X}| \leq |X| - (\ell + |B| + 1) + |B| + \ell < |X|$, which contradicts optimality of $X$.
    \end{proof}
    
    \begin{rrule}\label{rr:dtmwc_high_deg_B}
        Assume some boundary vertex $v \in B$, such that there exist paths from $v$ to $\ell + |B| + 2$ many $T_G$-terminals, and these paths only intersect at $v$. Then remove all edges incident with $v$, add to $T_G$ two new terminals $t_{v, 1}, t_{v, 2}$ and make them adjacent to $v$.
    \end{rrule}
    \begin{lemma}
        Reduction Rule~\ref{rr:dtmwc_high_deg_B} is gluing safe.
    \end{lemma}
    \begin{proof}
        Here, the idea is the same as in the proof for gluing safeness of Reduction Rule~\ref{rr:dtmwc_high_deg_F}. The difference is that instead of directly putting $v$ into the solution, we remember that $v$ needs to be in the solution. Again, fix any graph $H_B$ and terminal set $T_H$, let $Q = G_B \oplus H_B$, $T = T_G \cup T_H$, as well as $G'$ the resulting graph from $G$ by applying the reduction rule on vertex $v$, let $Q' = G'_B \oplus H_B$, $T'_G = T_G \cup \{t_{v, 1}, t_{v, 2}\}$, and $T' = T'_G \cup T_H$.
        
        ``$\OPT(Q, T) \leq \OPT(Q', T') + 1$'':
        Let $X'$ be a feasible solution for $(Q', T')$. Without loss of generality we can assume that $v$ is contained in $X'$: Otherwise, it holds that at least one of $t_{v, 1}, t_{v, 2}$ is contained in $X'$ and thus $(X' \setminus \{t_{v, 1}, t_{v, 2}\}) \cup \{v\}$ would be a feasible solution of at most the same size. As it holds that $Q - v = Q' - \{v, t_{v, 1}, t_{v, 2}\}$, it follows that $X'$ is also a feasible solution for $(Q, T)$.
        
        ``$\OPT(Q, T) \geq \OPT(Q', T') + 1$'':
        Let $X$ be an optimum solution for $(Q, T)$. By the argumentation in the same direction of the proof for Reduction Rule~\ref{rr:dtmwc_high_deg_F}, we can assume that $v$ is contained in $X$. By the equality $Q - v = Q' - \{v, t_{v, 1}, t_{v, 2}\}$ and the fact that terminals $t_{v, 1}, t_{v, 2}$ of $T'$ are only adjacent to $v$ in $Q'$, it thus follows that $X$ is also feasible for $(Q', T')$.
    \end{proof}
    
    \begin{lemma}\label{lem:dtmwc_T_S}
        If neither of Reduction Rules~\ref{rr:dtmwc_iso}, \ref{rr:dtmwc_high_deg_F}, and \ref{rr:dtmwc_high_deg_B} can be applied, then it holds that $|T_G| \in \Oh((|B|+\ell)^2)$.
    \end{lemma}
    \begin{proof}
        Fix $H_B$ and $T_H$, assume optimal global solution $X$. Let $Q$ be the resulting global graph $G_B \oplus H_B$ and $T$ the total terminal set $T_G \cup T_H$. Obviously, there are at most $|B|$-many components of $Q-X$ that intersect $B$. As by assumption we cannot apply Reduction Rule~\ref{rr:dtmwc_iso}, it follows that all other components of $Q-X$ that contain $T_G$-terminals, must be adjacent to $X \cap V(G)$. Since also Reduction Rules~\ref{rr:dtmwc_high_deg_F} and \ref{rr:dtmwc_high_deg_B} cannot be applied, each vertex from $X \cap V(G)$ can be adjacent to at most $\ell + |B| + 1$-many components of $Q-X$ that do not intersect $B$ and contain $T_G$-terminals. 
        
        Altogether, there are: (i) at most $\ell + |B|$-many $T_G$-terminals that are contained in $X$ (as in the proof for gluing safeness of Reduction Rule~\ref{rr:dtmwc_high_deg_F}, we can assume that $|X \cap V(G)|$ is upper bounded by $\ell + |B|$); (ii) at most $|B|$-many $T_G$-terminals contained in the components of $Q - X$ that intersect $B$; and (iii) at most $(\ell + |B|)(\ell + |B| + 1)$-many $T_G$-terminals whose corresponding components in $Q-X$ do not intersect $B$ and are adjacent to $X \cap V(G)$. Summing up, we get at most $(\ell + |B|)(\ell + |B| + 3) \in \Oh((|B| + \ell)^2)$ terminals contained in $T_G$.
    \end{proof}
    
    We are now ready to reduce the vertex set to size $\Oh(|B|+|T_G|)^3)$.
    Let $G^*$ be the mixed graph obtained from $G$ by adding vertices $v', v''$ for each $v \in F = V(G) \setminus B$, and edges $(u, v'), (u, v'')$ for each edge $\{u, v\} \in E(G)$, i.e., $v'$ and $v''$ are \emph{sink-only copies} of $v$ in $G^*$. 
    
    \begin{rrule}\label{rr:dtmwc}
        Let $\mathcal{M}$ be the gammoid for $G^*$ with sources $B \cup T_G$ and let $\mathcal{Z} := \{\{v, v', v''\} \mid v \in F\}$. Further, let $\mathcal{Z}'$ be a $(2|B|+|T_G|-1)$-representative subset of $\mathcal{Z}$. Then bypass every vertex $v \in F \setminus T_G$ for which $\{v, v', v''\}$ is not contained in $\mathcal{Z}'$.
    \end{rrule}
    
    \begin{lemma}\label{lem:dtmwc}
        Reduction Rule~\ref{rr:dtmwc} is gluing safe.
    \end{lemma}
    \begin{proof}
        Let $G'$ be the resulting graph. Fix $H_B, T_H$. Let $T := T_G \cup T_H$.
        
        ``$\OPT(G_B \oplus H_B, T) \leq \OPT(G'_B \oplus H_B, T)$'':
        Let $X'$ be a solution for $(G'_B \oplus H_B, T)$. Obviously, $X'$ does not contain any bypassed vertices and thus by Lemma~\ref{lem:bypass} it is also a solution for $(G_B \oplus H_B, T)$.
        
        ``$\OPT(G_B \oplus H_B, T) \geq \OPT(G'_B \oplus H_B, T)$'':
        Let $X$ be a solution for $Q := G_B \oplus H_B$ with terminal set $T$, such that $|X|$ is minimized but among such $|X \cap T|$ is maximized. Let $X_F := X \cap F$, $X_0 := X \setminus T$, and $T_0 := T \setminus X$. Further, define the mixed graph $Q^*$ on the base of $Q$, but with two additional sink-only copies $v', v''$ for every vertex $v \in V(Q)$. 
        Kratsch~\cite{Kratsch2014-RecentDevelopmentsKernelization} has shown that in $Q^* - (T \cap X)$ there exists a path packing with $|X_0| + 2$ paths from $T_0$ to $X_0 \cup \{x', x''\}$ for each $x \in X_0$. Note that this implies existence of a path packing in $G^* - (T \cap X)$ with $|X_0 \cap F| + 2$ paths from $(B \cup T_G) \setminus X$ to $(X_0 \cap F) \cup \{x', x''\}$ for each $x \in X_0 \cap F$: Take the path packing from $T_0$ to $X_0 \cup \{x', x''\}$ in $Q^* - (T \cap X)$, throw away any path going to $X_0 \setminus F$, and note that the remaining paths are either totally contained in $G^*[F \cup \{x', x''\}]$ or must intersect $B \setminus X$. For each path of the latter case we fix its last vertex $u$ contained in $B \setminus X$ and use the subpath from $u$ to $(X_0 \cap F) \cup \{x', x''\}$ for the path packing.
        
        Now we show for any fixed $x \in X_0 \cap F$ that $\{x, x', x''\}$ is contained in $\mathcal{Z}'$. Fix the path packing of size $|X_0 \cap F| + 2$ from $(B \cup T_G) \setminus X$ to $(X_0 \cap F) \cup \{x', x''\}$ in $G^* - (T \cap X)$ and let $U$ be the set of $B \setminus X$-vertices from which the paths go to $(X_0 \cap F) \cup \{x', x''\}$ without touching $B$ any further. Set $Y := B \setminus U$ and observe that there exists a path packing of $|X_0 \cap F| + |Y| + 2$ paths from $B \cup T_G$ to $(X_0 \cap F) \cup Y \cup \{x', x''\}$ (including the length-$0$ path $(y)$ for each $y \in Y$) in $G^*$, i.e., $\{x, x', x''\}$ extends $X_x := ((X_0 \cap F) - \{x\}) \cup Y$ in $\mathcal{M}$. Let us make sure that $|X_x|$ is upper bounded by $2|B| + |T_G| - 1$, thus leading to the fact that $\{x, x', x''\}$ is a correct candidate for $\mathcal{Z}'$: If it were true that $|X_x| > 2|B| + |T_G|-1$, it would follow thereof that $|X_F| > |B| + |T_G|$ and thus that $\tilde{X} := (X \setminus V(G)) \cup B \cup T_G$ is a multiway cut for $Q$ of smaller size, contradicting minimality of $|X|$. For this last point we remark that a potential path between terminals in $Q - \tilde{X}$ would need to be a path between terminals from $T_H \setminus B$ and contained entirely inside of $H - X$, contradicting the fact that $X$ is a multiway cut for $(Q, T)$.
        
        Having seen that $\{x, x', x''\}$ is a candidate for $\mathcal{Z}'$ as an extension of $X_x$, we show that there is no other vertex $v \in F$ which extends $X_x$. Assume for contradiction that such a vertex $v$ exists. First for any $v \in X_0 \setminus \{x\}$ it holds that $\{v, v', v''\} \cap X_x = \{v\} \neq \emptyset$ and thus that $\{v, v', v''\}$ does not extend $X_x$. Hence we consider the case that $v \in F \setminus X$. We remark that by choice of $U$ we have that $U \cap X = \emptyset$, and that each of its vertices $u \in U$ is reachable from some terminal $t \in T$ in $Q - X$, with no two $U$-vertices sharing the terminal. Since by assumption there exist three paths from $B \cup T_G$ to $\{v, v', v''\}$ and at most one of them goes through $x$, two of these paths go directly through $G - X_x$. In particular this means that these two paths go from $U \cup T_G$ to $\{v, v', v''\}$, implying that $v$ is connected with two different terminals in $Q - X$, a contradiction to $X$ being a multiway cut for $(Q, T)$.
        
        As a result, for any fixed $x \in X_0 \cap F$ it must hold that $\{x, x', x''\}$ is contained in $\mathcal{Z}'$ and thus no $X$-vertices were bypassed while constructing $G'$. By Lemma~\ref{lem:bypass} this means that $X$ is also a multiway cut for $(G'_B \oplus H_B, T)$.
    \end{proof}
    
    \begin{proof}[Proof of Theorem~\ref{thm:pbk_dtmwc}]
        We are given a boundaried graph $G_B$ and set $T_G$ of at most $\ell$-many local terminals. Apply Reduction Rules~\ref{rr:dtmwc_iso}, \ref{rr:dtmwc_high_deg_F}, and \ref{rr:dtmwc_high_deg_B} exhaustively, and remember that by Lemma \ref{lem:dtmwc_T_S} it holds that $|T_G| \in \Oh(|B| + \ell)^2$. Next, using Lemma~\ref{lem:repset} we compute in randomized polynomial time the $(2|B| + |T_G| - 1)$-representative subset $\mathcal{Z}'$ of $\mathcal{Z} = \{\{v, v', v''\} \mid v \in F\}$ with respect to the gammoid $\mathcal{M}$ for $G^*$ with sources $B \cup T_G$, which we then apply Reduction Rule~\ref{rr:dtmwc} on. Note that $\mathcal{Z}'$ has size at most $\binom{3+(2|B| + |T_G| - 1)}{3} \in \Oh((|B| + |T_G|)^3)$. Since we have bypassed all vertices in $F \setminus T_G$ for which $\{v, v', v''\}$ is not contained in $\mathcal{Z}'$, the gluing equivalent (by Lemma~\ref{lem:dtmwc}) boundaried graph $G'_B$ has $\Oh((|B| + |T_G|)^3) = \Oh((|B| + \ell)^6)$ vertices, which we output together with the terminal set $T_G$.
    \end{proof}
    
    \section{Odd Cycle Transversal[local]}\label{section:oct}
    While in the previous sections the representative sets framework on gammoids was used in order to obtain boundaried kernelizations for problems, which are defined through cuts, we will now also apply it for the problem \oddcycletransversal, where we only want to hit every odd cycle, but otherwise connectedness is still allowed for the remaining graph. Reed et al.\ \cite{ReedEtAl2004-FindingOddCycle} have used a connection between odd cycle transversals and cuts, and this connection is what allows us to use representative sets on gammoids. The main part is the following: Assume that we are given some odd cycle transversal $X$ (not necessarily optimal) for graph $G$. Let $G^*$ be an auxiliary graph with dedicated vertices $X^*$ corresponding to $X$-vertices. Then an optimal solution for $G$ can be deduced by computing for each partition $S \cup T \cup R$ of $X^*$ a minimum $(S, T)$-cut in $G^*-R$. 
    
    Kratsch and Wahlstr\"om \cite{KratschWahlstrom2014-CompressionMatroids} used this connection in order to get a polynomial compression in the form of a matroid with $\Oh(|X|)$ elements and randomized encoding size of $\Oh(|X|^3)$, which includes information on the min-cuts for each partition $S \cup T \cup R$ of $X^*$. Later the same authors showed a randomized polynomial kernelization for the problem \parameterizedproblem{$\Gamma$-Feedback Vertex Set}{$k$} \cite{KratschWahlstrom2020-RepresentativeSetsIrrelevantVertices}, which includes \parameterizedproblem{Odd Cycle Transversal}{$k$} as a special case. A written out kernel specifically for \parameterizedproblem{Odd Cycle Transversal}{$k$} can be found in \cite{Fomin_Lokshtanov_Saurabh_Zehavi_2019}.
    
    \begin{theorem}\label{thm:pbk_oct}
        The parameterized problem \oddcycletransversallocal admits a randomized polynomial boundaried kernel with $\Oh((|B| + \ell)^6)$ vertices.
    \end{theorem}
    
    Let us be given boundaried graph $G_B$ and odd cycle transversal $S$ for $G$. By Lemma~\ref{lem:bk_superset} and since the class of bipartite graphs is hereditary, we can simply assume that $S$ is a subset of $B$ and redefine $B := B \cup S$ otherwise. Let $F := V(G) \setminus B$ and note that by previous assumption it holds that $G[F] = G-B$ is a bipartite graph.
    Fix some arbitrary bi-coloring $(L, R)$ of $G[F]$ and define an auxiliary graph $G^*$ with vertices $B_L \cup B_R \cup F$, where $B_\side := \{x_\side \mid x \in B\}$ for $\side \in \{L, R\}$. The graph $G^*$ contains edges such that $G^*[F] = G[F]$ and for all $x \in B$ vertex $x_L$ is adjacent to $N_G(x) \cap R$ and $x_R$ is adjacent to $N_G(x) \cap L$. In addition, for any pair $y, x \in B$ that are adjacent in $G$, add the edges $\{y_L, x_R\}$ and $\{y_R, x_L\}$ to $G^*$.
    
    Assuming to be given some graph $H_B$, we show that any optimum solution for $G_B \oplus H_B$ corresponds to the combination of an odd cycle transversal $X$ of $H$ and a minimum cut between some vertex sets $B_\sim, B_\Delta \subseteq B$ in $G^* - D$, where $B_\sim \cup B_\Delta \cup D$ is a partition of $B_L \cup B_R$ that is compatible with the local odd cycle transversal $X$ of $H$. Afterwards, we use the gammoid on $G^*$ with sources $B_L \cup B_R$ and its cut cover $Z$ by Lemma~\ref{lem:cutset} in order to bypass all vertices in $F \setminus Z$ and subdivide created paths of length one, preserving bipartiteness of $G-B$.
    
    \begin{lemma}\label{lem:oct}
        Let $G_B$ and $G^*$ be as above and let $H_B$ be one more boundaried graph. The size of a minimum odd cycle transversal for $G_B \oplus H_B$ is equal to the minimum over the sum of the sizes of an odd cycle transversal $X$ for $H$ and of a minimum cut between $B_\sim := (B_L \cap K^*_L) \cup (B_R \cap K^*_R)$ and $B_\Delta := (B_L \cap K^*_R) \cup (B_R \cap K^*_L)$ in $G^* - \{x_L, x_R \mid x \in B \cap X\}$, where $(K_L, K_R)$ is an arbitrary bi-coloring of $H-X$ and $K^*_\side := \{x_L, x_R \mid x \in B \cap K_\side\}$ for $\side \in \{L, R\}$.
    \end{lemma}
    \begin{proof}
        ``$\geq$'':
        Let $S$ be an odd cycle transversal for $G_B \oplus H_B$. Set $X = S \cap V(H), C = S \cap F$ and choose arbitrary bipartition $(K_L, K_R)$ of $H-X$. Assume for contradiction that $C$ is not a cut between $B_\sim$ and $B_\Delta$, i.e., there exists a path between some vertices $u \in B_\sim$ and $v \in B_\Delta$ in $G^* - (C \cup \{x_L, x_R \mid x \in B \cap X\})$.
        First consider, without loss of generality, the case that $u \in \{B_L \cap K^*_L\}$ and $v \in \{B_L \cap K^*_R\}$, and hence also that $u = x_L, v = y_L$ for some $x \in B \cap K_L, y \in B \cap K_R$. Since both $u$ and $v$ are of the same color $B_L$, the path between them in $G^* - (C \cup \{x_L, x_R \mid x \in B \cap X\})$ has even length. By choice of $C$ the same path also exists in $(G_B \oplus H_B - S)$ (after replacing $u, v$ by $x, y$), but by assumption we have given different colors to $x$ and $y$, which leads to a contradiction.
        Next consider the case that $u \in \{B_L \cap K^*_L\}$ and $v \in \{B_R \cap K^*_L\}$, i.e., $u = x_L, v = y_R$ for some $x, y \in B \cap K_L$. Again, the same path exists in $(G_B \oplus H_B - S)$, but this time $x, y$ have the same color while the path has odd length, again a contradiction.
        
        ``$\leq$'':
        Let $X$ be an odd cycle transversal for $H$, let $(K_L, K_R)$ be a bi-coloring of $V(H) \setminus X$ and let $C$ be a cut between $B_\sim$ and $B_\Delta$ in $G^* - \{x_L, x_R \mid x \in B \cap X\}$. Let $C^* := (F \cap C) \cup \{x \mid \{x_L, x_R\} \cap C \neq \emptyset\}$. We show that $S := X \cup C^*$ is an odd cycle transversal for $G_B \oplus H_B$. Assume for contradiction the existence of an odd cycle $Q$ in $(G_B \oplus H_B) - S$. This cycle needs to contain a $B$-vertex $x$, as $X \subseteq S$ is an odd cycle transversal for $H$ and $G[F]$ is bipartite. Consider the case that $Q$ contains exactly one $B$-vertex, i.e., the cycle is totally contained in $G$ (again, $H-X$ is bipartite). Assume without loss of generality that $x$ is contained in $K_L$. Due to $Q$, there exists a path between $x_L$ and $x_R$ in $G^* - (C \cup \{x_L, x_R \mid x \in B \cap X\})$, and thus also between $B_L \cap K^*_L$ and $B_R \cap K^*_L$, which is a contradiction to the fact the $C$ is a cut between $B_\sim$ and $B_\Delta$.
        Now consider the case that $Q$ contains at least two $B$-vertices, i.e., there exists some $y \in (B \cap Q) \setminus \{x\}$. Since $Q$ is an odd cycle, there exists an odd and an even path between $x$ and $y$. Now if $x,y$ are contained in the same color for $H-X$, w.l.o.g.\ $K_L$, then we get a contradiction through the analogue of the odd-length path of $Q$ between $x_L \in B_L \cap K^*_L \subseteq B_\sim$ and $y_R \in B_R \cap K^*_L \subseteq B_\Delta$.
        On the other hand, if $x$ and $y$ are contained in different colors for $H-X$, say, $x \in K_L$ and $y \in K_R$, then the contradiction arises due to the analogue of the even-length path of $Q$ between $x_L \in B_L \cap K^*_L \subseteq B_\sim$ and $y_L \in B_L \cap K^*_R \subseteq B_\Delta$.
    \end{proof}
    
    \begin{rrule}\label{rr:oct}
        Let $G_B, G^*$ and the corresponding vertex sets be as above. Let $Z$ be a cut cover for $G^*$ and the set $B_L \cup B_R$. Then remove all vertices in $F \setminus Z$ in $G$. For any two vertices $u, v$ in $B \cup (F \cap Z)$, if there was an odd-length path between $u$ and $v$ with all inner vertices in $F \setminus Z$, then add an edge between $u$ and $v$; and if there was such an even-length path between $u$ and $v$, then add a path of length two between $u$ and $v$. It might be the case that we add both an edge and a length-2 path between $u$ and $v$. Also, for any $x \in B$, if there was an odd path consisting only of $x$ and vertices from $F \setminus Z$, then add a triangle with $x$ one of its vertices.
    \end{rrule}
    \begin{lemma}
        Reduction Rule~\ref{rr:oct} is gluing-safe.
    \end{lemma}
    \begin{proof}
        Let $G'$ be the resulting graph. Observe that $G'-B$ is bipartite as we only added independent edges for the triangles, substituted odd-length paths by edges, and even-length paths by length-2 paths. Construct the graph $\tG$ from $G'$ in the same way as $G^*$ was constructed from $G$. The following claim is analogous to Lemma~\ref{lem:bypass} and used for our proof.
        
        \begin{claim}\label{claim:oct}
            For any set $C \subseteq B_L \cup B_R \cup (F \cap Z)$ and any vertices $s, t \in B_L \cup B_R$, there is a path from $s$ to $t$ in $G^* - C$ if and only if there is a path from $s$ to $t$ in $\tG - C$.
        \end{claim}
        \begin{claimproof}
            For the first direction, assume that $P = (v_1, \dots, v_r)$ is a path from $s = v_1$ to $t = v_r$ in $G^* - C$. If $P$ only contains vertices from $Z' := B_L \cup B_R \cup (F \cap Z)$, then obviously we are done. Hence consider any sequence $p, p+1, \dots, q-1, q$ such that (i) $1 < p \leq q < r$; (ii) each among $v_p, v_{p+1}, \dots, v_{q-1}, v_{q}$ is contained in $F \setminus Z$; and (iii) $v_{p-1}, v_{q+1}$ are contained in $Z'$. If w.l.o.g.\ it holds that $v_{p-1} = x_L, v_{q+1} = x_R$ for some $x \in B$, then we added a length-3 cycle $x, x_1, x_2$ to $G'$, and hence there exists the path $v_{p-1}, x_1, x_2, v_{q+1}$ in $\tG$. Else let $u$ be equal to $v_{p-1}$ if $v_{p-1} \in F$ and equal to $x$ if $v_{p-1} \in \{x_L, x_R\}$. Define $w$ in the same way for $v_{q+1}$. Observe that we have added in $G'$, between $u$ and $w$: an edge if $q-p$ is odd, and a path of length two, if $q-p$ is even. Therefore, by construction of $\tG$, there is either an edge or a length-2 path between $v_{p-1}$ and $v_{q+1}$ in $\tG$. 
            Either way we have thus shown that the subpath $v_p, v_{p+1}, \dots, v_{q-1}, v_{q}$ of $P$ in $G^*$ can been substituted by a path between $v_p$ and $v_q$ in $\tG$. As we can apply this argumentation on each such sequence of $F \cap Z$-vertices along the path $P$, and since no vertices used for these substituting paths are contained in $C \cup (F \setminus Z)$, we know that there exists a path between $s$ and $t$ in $\tG - C$.
            
            Now to the other direction. Assume that $P' = (v_1, \dots, v_r)$ is a path from $s = v_1$ to $t = v_r$ in $\tG - C$. If all edges in $P'$ can also be found in $G^*$, we are done. Thus consider any: (i) edge $\{v_i, v_{i+1}\} \not\in E(G^*)$ with $i \in [r-1]$; (ii) path $(v_i, v_{i+1}, v_{i+2})$ with $v_i, v_{i+2} \in V(G^*)$ but $v_{i+1} \not\in V(G^*)$; and (iii) path $(v_i, v_{i+1}, v_{i+2}, v_{i+3})$ with $\{v_i, v_{i+3}\} = \{x_L, x_R\}$ and $v_{i+1}, v_{i+2} \not\in V(G^*)$ for $i \in [r-3]\}$ and some $x \in B$.
            In cases (i) and (ii) there needed to exist in $G^*$ a path of odd or even length between $v_i$ and $v_{i+1}$ (for (i)), respectively $v_{i+2}$ (for (ii)). In case (iii) there was an odd cycle consisting of $x$ and vertices from $F \setminus Z$ in $G'$. Then there also exists an odd-length path between $x_L$ and $x_R$ in $G^*$. Note that none of these paths contain inner vertices from $B_L \cup B_R \cup (F \cap Z)$. Substituting these paths accordingly, we get a walk from $s$ to $t$ in $G^* - C$, which also implies such a path.
        \end{claimproof}
        
        ``$\OPT(G_B \oplus H_B) \leq \OPT(G'_B \oplus H_B)$'':
        Let $S'$ be an odd cycle transversal for $G'_B \oplus H_B$. By Lemma~\ref{lem:oct} we can divide this solution into an odd cycle transversal $X'$ for $H$ and a minimum size cut $C'$ between $B_\sim$ and $B_\Delta$ in $\tG - \{x_L, x_R \mid x \in B \cap X'\}$. We can assume without loss of generality that none of the vertices created for the substituting paths and cycles are contained in $C'$. Hence $C'$ only contains vertices from $B_L \cup B_R \cup (F \cap Z)$ and it follows from Claim~\ref{claim:oct} that $C'$ is also a cut between $B_\sim$ and $B_\Delta$ in $G^* - \{x_L, x_R \mid x \in B \cap X'\}$. Hence by Lemma~\ref{lem:oct} there exists an odd cycle transversal of size at most $|S'|$ for $G_B \oplus H_B$.
        
        ``$\OPT(G_B \oplus H_B) \geq \OPT(G'_B \oplus H_B)$'':
        Now let $S$ be an odd cycle transversal for $G_B \oplus H_B$ and divide $S$ by Lemma~\ref{lem:oct} into an odd cycle transversal $X$ for $H$ and a minimum size cut $C$ between $B_\sim$ and $B_\Delta$ in $G^* - \{x_L, x_R \mid x \in B \cap X\}$. Since $Z \subseteq V(G')$ contains a minimum cut for any bipartition of $(B_L \cup B_R) \setminus \{x_L, x_R \mid x \in B \cap X\}$ (and $(B_\sim, B_\Delta)$ is such a bipartition), we can assume without loss of generality that $C$ only contains $Z$-vertices. Then by Claim~\ref{claim:oct} it holds that $C$ is also a cut between $B_\sim$ and $B_\Delta$ in $\tG - \{x_L, x_R \mid x \in B \cap X\}$, and we can again apply Lemma~\ref{lem:oct} in order to get an odd cycle transversal of size at most $|S|$ for $G'_B \oplus H_B$.
    \end{proof}
    
    \begin{proof}[Proof of Theorem~\ref{thm:pbk_oct}]
        We get as input a boundaried graph $G_B$ and a local odd cycle transversal $S$ of size at most $\ell$. Let $(L, R)$ be any bi-coloring of $G - (B \cup S)$. We construct a graph $G'$ such that $G'_{B \cup S}$ is gluing-equivalent to $G_{B \cup S}$, and hence by Lemma~\ref{lem:bk_superset} such that $G'_B \equiv G_B$. For that, we construct the auxiliary graph $G^*$ as described above, but for $B' := B \cup S$ as the boundary instead of only $B$ and apply Lemma~\ref{lem:cutset} in order to compute in randomized polynomial time a set $Z \subseteq V(G^*)$ of size $\Oh((|B| + \ell)^3)$ which contains a minimum cut between any two subsets of $B'_L \cup B'_R$. We then apply Reduction Rule~\ref{rr:oct}. The resulting graph $G'$ has at most $\Oh((|B| + \ell)^6)$ vertices in $G'$: each of the vertices in $B \cup (F \cap Z)$, at most two vertices for a substituting triangle for each $x \in B$, and at most one additional vertex for a substituting length-2 path for any distinct pair $u, v \in B \cup (F \cap Z)$.
    \end{proof}
    
    \section{Vertex Cover[oct]}\label{section:vc}
    In this chapter we work with the parameterized problem \vertexcoveroct and we will combine two concepts for this:
    First, the idea of Bodlaender and Jansen \cite{JansenBodlaender2013-VertexCoverKernelizationRevisited} of reducing to a so-called \emph{clean instance}, i.e., such that $G-B$ has a perfect matching and hence by König's Theorem a minimum vertex cover for $G-B$ contains exactly half of the $F$-vertices (with $F := V(G) \setminus B)$. This allows to perform further reduction rules based on the increase of local solution size for $G-B$, when a certain subset $B' \subseteq B$ is explicitly forbidden, which is called the \emph{conflict} caused by $B'$ on $F$. We then construct an auxiliary graph $G^*$, for which we show that conflicts in $G$ correspond to certain cuts in $G^*$. This is where the representative sets framework of Kratsch and Wahlstr\"om comes in and allows us to find a cut cover in $G^*$, using which we choose vertices of $G$ to bypass. 
    
    \begin{theorem}\label{thm:pbk_vc_oct}
        The parameterized problem \vertexcoveroct admits a polynomial boundaried kernelization with $\Oh((|B| + \ell)^3)$ vertices.
    \end{theorem}
    
    In analogy to the previous section, we assume to be given boundaried graph $G_B$ such that $G-B$ is a bipartite graph. Set $F := V(G) \setminus B$. Apply a few reduction rules in order to get a clean instance.
    
    \begin{rrule}\label{rr:vc_iso}
        If $v \in F$ is an isolated vertex, remove it.
    \end{rrule}
    \begin{rrule}\label{rr:vc_crown}
        If $G$ contains a crown $(I, H)$ with $I \subseteq F$, an $H$-saturating matching $M$ and $M \subsetneq E(I, H)$, then remove all edges in $E(I, H) \setminus M$.
    \end{rrule}
    
    \begin{lemma}[\cite{AntipovKratsch2025-BoundariedKernelization}]\label{lem:vc_pre_poly}
        Reduction Rules~\ref{rr:vc_iso} and \ref{rr:vc_crown} are gluing-safe and can be applied exhaustively in polynomial time.
    \end{lemma}
    
    \begin{lemma}[{\cite[Lemma 13]{AntipovKratsch2025-BoundariedKernelization}}]\label{lem:vc_pre_stable}
        Let $G_B$ be a boundaried graph on which Reduction Rules~\ref{rr:vc_iso} and \ref{rr:vc_crown} cannot be applied. Then it holds that $\alpha(G-B) > \frac{1}{2}|V(G)|$.
    \end{lemma}
    
    We make the instance clean by finding a maximum matching $M$ of $G[F]$ and moving the set $I$ of vertices not matched by $M$ to the boundary. Due to Lemma~\ref{lem:vc_pre_stable} and K\H{o}nig's Theorem, the set $I$ has size at most $|B|$ (for details, see the proof in \cite[Lemma 13]{AntipovKratsch2025-BoundariedKernelization}). The gluing safeness follows from Lemma~\ref{lem:bk_superset}. After applying the reduction rule, we still denote the new boundary simply as $B$ and $V(G) \setminus B$ as $F$.
    
    \begin{rrule}\label{rr:vc_clean}
        If Reduction Rules~\ref{rr:vc_iso} and \ref{rr:vc_crown} cannot be applied, then replace the boundary by a set $\hat{B} \supseteq B$ of size at most $2|B|$, such that $G-\hat{B}$ has a perfect matching.
    \end{rrule}
    
    Since $G-B$ is bipartite, we compute (in linear time using a simple vertex transversal) a bipartition of $F$ into the sets $L$ and $R$, such that $G[L]$ and $G[R]$ are independent sets. Furthermore, we compute (in polynomial time using the Hopcroft-Karp algorithm \cite{HopcroftKarp1973-$n^52Algorithm}) a maximum matching $M$ of $G[F]$, which is perfect, assuming that we have applied Reduction Rules~\ref{rr:vc_iso}, \ref{rr:vc_crown}, and \ref{rr:vc_clean} beforehand.
    
    Based on $G_B$ and the chosen bipartition, we construct the auxiliary graph $G^*$ (with directed edges), in which each $B$-vertex $v$ is replaced by a left copy $v_L$ that has outgoing edges to $N_G(v) \cap R$  and a right copy $v_R$ that has incoming edges from $N_G(v) \cap L$. In addition, for each adjacent vertex pair $u, v \in B$ in $G$, we create the edges $(u_L, v_R)$ and $(v_L, u_R)$ in $G^*$.
    For any $B^* \subseteq B$ we denote the set of all left and right copies of $B^*$-vertices by $B^*_L$ and $B^*_R$, respectively. The edges between $L$ and $R$ are copied from $G$, but directed from the $L$-vertex to the $R$-vertex. Additionally, edges from $M$ are also directed in the other way.
    We will show that so-called conflicts in $G$ are tightly connected to corresponding cuts in $G^*$. So let us give the definition of conflicts and show some useful behavior of the vertex cover problem on graphs $G_B \oplus H_B$ with separator $B$.
    
    \begin{definition}[Conflict]
        Let $G$ be a graph and $F, B$ two disjoint subsets of $V(G)$. For any $B' \subseteq B$ we denote by $\conf{G,F}{B'} := \OPT_\VC(G[F] - N_{G, F}(B')) + |N_{G, F}(B')| - \OPT_\VC(G[F])$ the conflict caused by $B'$ on $F$.
    \end{definition}
    \begin{definition}[Cut]
        Let $G$ be a graph and $S, T \subseteq V(G)$ two vertex sets of $G$. We denote the size of a minimum cut on $G$ and $(S, T)$ by $\cut_G(S, T)$.
    \end{definition}
    
    \begin{lemma}\label{obs:vc_destruct}
        Let $G_B$ and $H_B$ be boundaried graphs.
        \begin{enumerate}
            \item[(i)] For any subset $B' \subseteq B$ and vertex covers $S_G, S_H$ for $G-B'$ and $H-B'$ respectively, the set $S := B' \cup S_G \cup S_H$ is a vertex cover for $G_B \oplus H_B$. 
            \item[(ii)] For any vertex cover $S$ for $G_B \oplus H_B$ and $B' := B \cap S$ the sets $S_G := S \cap (V(G) \setminus B)$ and $S_H := S \cap (V(H) \setminus B)$ are vertex covers for $G-B'$ and $H-B'$, respectively.
            \item[(iii)] It holds that $\OPT_\VC(G_B \oplus H_B) = \min_{B' \subseteq B} |B'| + \OPT_\VC(G-B') + \OPT_\VC(H-B')$.
        \end{enumerate}
    \end{lemma}
    \begin{proof}
        For (i) observe that $S$ covers all edges incident with $B'$ as well as all edges contained in $G-B'$ and in $H-B'$, and thus altogether all edges of $G_B \oplus H_B$.
        For (ii) it is clear that $S_G$ and $S_H$ cover all edges of $G-B$, respectively $H-B$. Furthermore, since $B \cap S = B'$ and $S$ is a vertex cover, there are no edges between vertices in $G[B \setminus B']$ or $H[B \setminus B']$, and as a result, $S_G$ and $S_H$ even cover all edges of $G-B'$, respectively $H-B'$.
        And as last for (iii) we consider each of the two directions. ``$\leq$'': Given the set $B' \subseteq B$ minimizing the right hand size and given minimum vertex covers $S_G, S_H$ for $G-B'$ and $H-B'$, we know by (i) that $\OPT_\VC(G_B \oplus H_B) \leq |B'| + |S_G| + |S_H|$. ``$\geq$'': Given a minimum vertex cover $S$ for $G_B \oplus H_B$ and setting $B' := B \cap S$, $S_G := S \cap (V(G) \setminus B')$, and $S_H := S \cap (V(H) \setminus B')$, we use (ii) to see that $\OPT_\VC(G-B') + \OPT_\VC(H-B') \leq |S_G| + |S_H| = |S| - |B'|$.
    \end{proof}
    
    Now to the correspondence of cuts in $G^*$ and conflicts in $G$. We show the connection separately for each of the two directions, and use them subsequently for our next reduction rule.
    
    \begin{lemma}\label{lem:vc_CtoS}
        Let $G$ be a graph and $B \subseteq V(G)$ such that $G-B$ has bipartition $(L, R)$ and a perfect matching $M$. Let $G^*$ be constructed from $G$ as described above. Let $B^* \subseteq B$ be an independent set, $B' := B \setminus B^*$ and $C$ a minimum cut between  $B^*_L$ and $B^*_R$ in $G^*$. Then there exists a vertex cover $S$ of size $|M| + |C|$ for $G-B'$ which contains every vertex $x \in B^*$ for which $C$ has a non-empty intersection with $\{x_L, x_R\}$.
    \end{lemma}
    \begin{proof}
        Define the set $S^*$ as the union of $C^* := C \cup N_M(C)$, $X_L := \{v \in L \setminus C^* \mid v \text{ reaches } B^*_R \text{ in } G^* \setminus C^*\}$, $X_R := \{v \in R \setminus C^* \mid v \text{ is reachable from } B^*_L \text{ in } G^* \setminus C^*\}$ and $Y_L := L \setminus (C^* \cup X_L \cup N_M(X_R))$. Based on $S^*$ we construct the set $S$, which contains all of $S^* \cap (L \cup R)$ and every vertex $x \in B^*$ for which $x_L$ and/or $x_R$ is contained in $C$.
        Let $Y_R := R \setminus (C^* \cup X_R \cup N_M(X_L))$.
        
        Note that $X_L$ has no $M$-neighbors in $X_R$ (and vice versa), since this would mean existence of an edge from $X_R$ to $X_L$, and thus a path from $B^*_L$ to $B^*_R$ in $G^* \setminus C^*$, a contradiction to $C \subseteq C^*$ being a cut between $B^*_L$ and $B^*_R$ in $G^*$. Clearly, also $Y_L$ does not contain $M$-neighbors of $X_R$ per definition. Hence, the only $M$-edges for which $S$ contains both endpoints are those between $C$ and $N_M(C)$, leading to $|S| \leq |S^*| \leq |M| + |C|$. Now let us make sure that $S$ is a vertex cover for $G-B'$. First, all neighbors of $B^* \setminus S = \{x \in B^* \mid C \cap \{x_L, x_R\} = \emptyset\}$ are contained in $C^*$, $X_L$, or $X_R$ and thus also in $S$. Next, similarly to the absence of edges from $X_R$ to $X_L$, there are no edges from $N_M(X_R)$ to $N_M(X_L)$, as those would lead to a path from $B^*_L$ through $X_R$ and $N_M(X_R)$ to $N_M(X_L)$ and thus through $X_L$ back to $B^*_R$, which would contradict $C \subseteq C^*$ being a cut between $B^*_L$ and $B^*_R$ in $G^*$. And last, an edge from $N_M(X_R)$ to $Y_R$ would violate the definition of $X_R$, as there would exist a path from $B^*_L$ through $X_R$ and $N_M(X_R)$ to $Y_R$, but all $R\setminus C^*$-vertices reachable from $B^*_L$ in $G^* \setminus C^*$ are contained in $X_R$. Altogether, $S$ is indeed a vertex cover for $G - B'$.
    \end{proof}
    
    \begin{lemma}\label{lem:vc_StoC}
        Let $G$ be a graph and $B \subseteq V(G)$ such that $G-B$ has bipartition $(L, R)$ and a perfect matching $M$. Let $G^*$ be constructed from $G$ as described above. Then for every set $B^* \subseteq B$ it holds that $\cut_{G^*}(B^*_L, B^*_R) \leq \conf{G, F}{B^*}$.
    \end{lemma}
    \begin{proof}
        Let $S$ be a vertex cover of $G[F]$ with $N_{G, F}(B^*) \subseteq S$ and minimum size among such. Clearly it then holds that $|S| = \OPT_\VC(G[F]) + \conf{G, F}{B^*}$. Since $G[F]$ has a perfect matching, it holds by K\H{o}nig's Theorem that $\OPT_\VC(G[F]) = \frac{1}{2}|F|$. Observe that for each edge $e \in M$, the vertex cover $S$ needs to contain at least one endpoint, but since it further contains  $\conf{G, F}{B^*}$-many additional vertices, and these are all matched by $M$, it follows that $S$ contains both endpoints of $\conf{G, F}{B^*}$-many $M$-edges. Let $M' \subseteq M$ be the set of these matching edges.
        
        Let $C \subseteq F$ be any set containing exactly one vertex from each $M'$-edge. We show that $C$ is a cut between $B^*_L$ and $B^*_R$ in $G^*$. Assume for contradiction the existence of some path $P$ in $G^*$ which goes from $B^*_L \setminus C$ to $B^*_R \setminus C$ and does not contain any edge from $M'$. Clearly, $P$ contains some vertices from $N_{G^*}(B^*_L) = N_{G, R}(B^*)$ and $N_{G^*}(B^*_R) = N_{G, L}(B^*)$, and without loss of generality we can assume that the path contains only one vertex from each of these sets. Let $P'$ be the subpath of $P$ without its endpoints, hence the endpoints of $P'$ lie in $N_{G, R}(B^*)$ and $N_{G, L}(B^*)$. Denote the vertices in $P'$ in an ordered manner, i.e., let $P' = (v_1, v_2, \dots v_r)$, with $v_1 \in N_{G, R}(B^*)$ and $v_r \in N_{G, L}(B^*)$. Since $P'$ is a path in $G^*[F]$, which itself is a directed version of $G[F]$, starts with a vertex from $R$ and ends with one from $L$, it holds that $P'$ has odd length, i.e., $r$ is even, and the edge $\{v_i, v_{i+1}\}$ is contained in $M$ for every odd $i \in [r-1]$. In an inductive manner, we show that for any $i \in [r]$ it holds that $S \cap \{v_1, \dots, v_i\} = \{v_j \mid j \text{ is odd and } j \leq i\}$. The base case is simple, since $v_1 \in N_{G, F}(B^*) \subseteq S$. For the inductive step, let $i \in \{2, \dots, r\}$ and assume that $S \cap \{v_1, \dots, v_{i-1}\} = \{v_j \mid j \text{ is odd and } j \leq i-1\}$. If $i$ is odd, then $v_{i-1}$ is not contained in $S$, and in order to cover the edge $\{v_{i-1}, v_i\}$, it must hold that $v_i \in S$; else $i$ is even, and since $v_{i-1} \in S$ by induction hypothesis, while the edge $\{v_{i-1}, v_i\}$ is contained in $M$, it cannot hold that $v_i \in S$: otherwise both endpoints of $\{v_{i-1}, v_i\} \in M$ are contained in $S$, and thus $\{v_{i-1}, v_i\} \in M'$, contradicting the choice of $P'$. 
        As a result, it holds that $S \cap V(P') = \{v_j \mid j \text{ is odd and } j \in [r]\}$ which contradicts the fact that $r$ is even, while $v_r \in N_{G, F}(B^*) \subseteq S$. Hence no path $P$ exists in $G^*$, which goes from $B^*_L \setminus C$ to $B^*_R \setminus C$ and does not contain $M'$-edges; and thus $C$ is a valid cut of size $|M'| = \conf{G, F}{B^*}$ between $B^*_L$ and $B^*_R$ in $G^*$.
    \end{proof}
    
    \begin{rrule}\label{rr:vc_bypass}
        Assume one cannot apply Reduction Rules~\ref{rr:vc_iso}, \ref{rr:vc_crown}, and \ref{rr:vc_clean} on $G_B$, and hence $G-B$ has bipartition $(L, R)$ and perfect matching $M$. Let $G^*$ be constructed from $G_B$ as described above and let $Z \subseteq V(G)$ be a cut cover for $(B_L, B_R)$ in $G^*$. If there exist vertices $u, v \in F \setminus Z$ that do not share any neighbors in $B$ and with $\{u, v\} \in M$, then bypass $u$ and $v$, and increase $\Delta$ by one.
    \end{rrule}
    
    \begin{lemma}
        Reduction Rule~\ref{rr:vc_bypass} is gluing-safe.
    \end{lemma}
    \begin{proof}
        Let $u, v$ be vertices as described in the reduction rule. Let $G'$ be the result of bypassing these vertices and $M' := M \setminus \{\{u, v\}\}$. Clearly, $M'$ is a perfect matching for $G' - B$, so let $\tG$ be the result of the same operation on $G'$ that was applied on $G$ in order to define $G^*$. 
        
        \begin{claim}\label{claim:vc_bypass}
            $\tG$ is equal to the result of bypassing $\{u, v\}$ in $G^*$.
        \end{claim}
        \begin{claimproof}
            Let $\hat{G}$ be the result of bypassing $\{u, v\}$ in $G^*$. We show the equality $\hat{G} = \tG$. The vertex set of both is equal to $B_L \cup B_R \cup F \setminus \{u, v\}$, so let us show equivalence of the edge sets.
            
            Clearly, if an edge $\{x, y\}$ exists in $G$ with $x, y \not\in \{u, v\}$, then also $\{x, y\} \in G'$ and the corresponding edge with the corresponding direction (from $B_L \cup L$ to $B_R \cup R$ and the other way around if $\{x, y\} \in M$) exists in both $\hat{G}$ and $\tG$. Next, assume existence of path $(x, u, v, y)$ in $G$ and assume without loss of generality that $u \in L, v \in R$. Note that by assumption of the reduction rule it holds that $x \neq y$, as $u$ and $v$ do not share neighbors in $B$. Observe that if $x$ is contained in $F$, then it is actually coming from $R$, as it has the neighbor $u \in L$, and the other way around, we know that $y$ is contained in $B \cup L$, since it holds that $v \in R$. In order to omit a case distinction, we will simply write $\hat{x}_R$ and mean $x_R$ if $x \in B$ and otherwise we will mean $x$ by it. And analogously we will write $\hat{y}_L$ for $y_L$ or $y$, respectively. By construction, there exist the edges $(u, \hat{x}_R), (\hat{y}_L, v)$ in $G^*$. On the other hand, in $G'$ exists the edge $\{x, y\}$ with $x \in B \cup R$ and $y \in B \cup L$, since we bypassed $u$ and $v$. Now if we bypass $u$ and $v$ in $G^*$, we get the edge $(\hat{y}_L, \hat{x}_R)$, and the same happens in $G'$ when we split $B$-vertices into $L$- and $R$-variants, give the edges the direction from $B_L \cup L$ to $B_R \cup R$ and the other way around if they are contained in $M$. In other words, the edge $(\hat{y}_L, \hat{x}_R)$ exists in both $\tG$ and $\hat{G}$. As a last step, note that any edge in both $\hat{G}$ and $\tG$ is created by one of the two cases mentioned above, i.e., since the (undirected variant of the) edge already exists in $G$, or because there is a path between the endpoints going through $u$ and $v$.
        \end{claimproof}
        
        ``$\OPT_\VC(G_B \oplus H_B) \leq \OPT_\VC(G'_B \oplus H_B) + 1$'':
        Let $S'$ be a solution for $G'_B \oplus H_B$ and $B' := B \cap S'$. By Lemma~\ref{obs:vc_destruct}(ii) the set $S'_H := S' \cap (V(H) \setminus B)$ is a vertex cover for $H - B'$ and $S'_{G'} := S' \cap (V(G') \setminus B)$ is a vertex cover for $G' - B'$. Let $B^* := B \setminus S'$, which is an independent set in $G'$. From Lemma~\ref{lem:vc_StoC} it follows that $\cut_{\tG}(B^*_L, B^*_R) \leq \conf{G', F}{B^*} \leq |S'_{G'}| - |M'|$. Because of Claim~\ref{claim:vc_bypass} and Lemma~\ref{lem:bypass} there exists a cut $C$ of size at most $|S'_{G'}| - |M'|$ between $B^*_L$ and $B^*_R$ in $G^*$. And finally, Lemma~\ref{lem:vc_CtoS} ensures existence of a vertex cover $S_G$ for $G-B'$ of size $|M| + |C| \leq 1 + |S'_{G'}|$. By Lemma~\ref{obs:vc_destruct}(i) this leads us to the vertex cover $S := S'_H \cup B' \cup S_G$ for $G_B \oplus H_B$ of size at most $|S'_H| + |B'| + |S_G| \leq |S'| + 1$.
        
        ``$\OPT_\VC(G_B \oplus H_B) \geq \OPT_\VC(G'_B \oplus H_B) + 1$'':
        Let $S$ be a minimum size vertex cover for $G_B \oplus H_B$. Let $B' := B \cap S$ and let $B^*$ be the independent set $B \setminus S = B \setminus B'$ of $G_B \oplus H_B$. Similar to the argumentation for the previous direction, we follow that $S$ can be decomposed into the vertex cover $S_H := S \cap (V(H) \setminus B)$ for $H-B'$, the set $B'$ and the vertex cover $S_G := S \cap (V(G) \setminus B)$ for $G-B'$. Let $C$ be a min-cut for $(B^*_L, B^*_R)$ in $G^*$. We can assume without loss of generality that $C$ only contains vertices from $Z$ and thus does not intersect $\{u, v\}$. By Lemma~\ref{lem:vc_StoC} we know that $|C| \leq |S_G| - |M|$. Now let $C'$ be a min-cut for $(B^*_L, B^*_R)$ in $\tG$, and note that by Claim~\ref{claim:vc_bypass} and Lemma~\ref{lem:bypass} it holds that $|C'| \leq |C| \leq |S_G| - |M'| - 1$. Applying Lemma~\ref{lem:vc_CtoS} on $C'$ we get vertex cover $S_{G'}$ of size $|M'| + |C'| = |S_G| - 1$ for $G'-B'$. By Lemma~\ref{obs:vc_destruct}(i) this leads us to the vertex cover $S' := S_H \cup B' \cup S_{G'}$ for $G'_B \oplus H_B$ of size at most $|S_H| + |B'| + |S_{G'}| \leq |S| - 1$.
    \end{proof}
    
    As a last step, it remains to bound the number of $M$-adjacent pairs $u, v$ that are not contained in $Z$, but share neighbors in $B$. Then we are ready to give the whole algorithm and size bound.
    
    \begin{rrule}\label{rr:vc_triangle}
        Assume one cannot apply Reduction Rules~\ref{rr:vc_iso}, \ref{rr:vc_crown}, and \ref{rr:vc_clean} on $G_B$, and hence $G-B$ has bipartition $(L, R)$ and perfect matching $M$.
        If there exists a vertex $x \in B$ with more than $|B|$ distinct pairs of vertices $u, v \in F$, such that $\{u, v\} \in M$ and $(u, v, x)$ forms a triangle, i.e., a cycle of length three, then remove all edges between $x$ and $F$, create a new vertex $l_x$, and add the edge $\{x, l_x\}$.
    \end{rrule}
    \begin{lemma}
        Reduction Rule~\ref{rr:vc_triangle} is gluing-safe.
    \end{lemma}
    \begin{proof}
        Let $G'$ be the resulting graph.
        
        ``$\OPT(G_B \oplus H_B) \leq \OPT(G'_B \oplus H_B)$'': Let $S'$ be a vertex cover for $G'_B \oplus H_B$. Due to $x$ being the only neighbor of $l_x$, we can assume that $x$ is contained in $S'$. But then $S'$ is also a solution for $G_B \oplus H_B$, since all edges of $E(G_B \oplus H_B) \setminus E(G'_B \oplus H_B)$ are incident with $x$.
        
        ``$\OPT(G_B \oplus H_B) \geq \OPT(G'_B \oplus H_B)$'': Now let $S$ be a minimum vertex cover for $G_B \oplus H_B$. Destruct the considered triangles in the following way: let $U$ be the vertices from $R$ and $W$ the vertices from $L$. Construct $G^*$ from $G_B$ as described above. For every pair $u \in U$, $v \in W$ such that $\{u, v\} \in M$ and $(u, v, x)$ is a triangle, it holds that $(x_L, u, v, x_R)$ is a path in $G^*$. Thus, by Lemma~\ref{lem:vc_StoC} it follows that $\conf{G, F}{B^*} \geq |B|$ for any $B^* \subseteq B$ that contains $x$. Together with Lemma~\ref{obs:vc_destruct}(iii) this implies that $S$, as an optimum solution, must contain $x$. From that, it again follows that $S$ is also a feasible solution for $G'_B \oplus H_B$, since all edges of $E(G'_B \oplus H_B) \setminus E(G_B \oplus H_B)$ (namely exactly the single edge $\{x, l_x\}$) are incident with $x \in S$.
    \end{proof}
    
    \begin{proof}[Proof of Theorem~\ref{thm:pbk_vc_oct}]
        We are given a boundaried graph $G_B$ and an odd cycle transversal $S$ of $G$. We define the new boundary $B' := B \cup S$ and apply Reduction Rules~\ref{rr:vc_iso}, \ref{rr:vc_crown}, and \ref{rr:vc_clean} in order to get a boundaried graph $\hat{G}_{\hat{B}}$ such that $G_B \equiv \hat{G}_{\hat{B}}$ and $\hat{G}-\hat{B}$ has bipartition $(L, R)$ and perfect matching $M$. Further apply Reduction Rule~\ref{rr:vc_triangle} in order to bound the number of $M$-adjacent pairs that share $\hat{B}$-neighbors by $|\hat{B}|^2$. Let $F = V(\hat{G}) \setminus \hat{B}$ and let $G^*$ be the directed graph obtained from $\hat{G}$ by splitting every $B$-vertex $x$ into $x_L$ and $x_R$ and adding edges from $x_L$ to the vertices in $N_{\hat{G}, R}(x) \cup \{y_R \mid y \in N_{\hat{G}, B}(x)\}$, edges from $N_{\hat{G}, L}(x) \cup \{y_L \mid y \in N_{\hat{G}, B}(x)\}$ to $x_R$, from every $v \in L$ to $N_{\hat{G}, R}(v)$ and from every $v \in R$ to $N_{M}(v)$. Then using Lemma~\ref{lem:cutset_small} we repeatedly compute a cut cover $Z$ for $(B_L, B_R)$ in $G^*$ and apply Reduction Rule~\ref{rr:vc_bypass} on some vertices $u, v \in F \setminus Z$ with $\{u, v\} \in M$ and without common neighbor in $\hat{B}$, recomputing $G^*$ at each step simply by bypassing $u, v$, as is ensured to be correct by Claim~\ref{claim:vc_bypass}. We stop when there is no edge $\{u, v\} \in M$ without common neighbor in $\hat{B}$ and with $\{u, v\} \cap Z = \emptyset$, implying that the size of $F$ is upper bounded by $2|\hat{B}|^2 + 2|Z| \in \Oh(|\hat{B}|^3) = \Oh((|B| + \ell)^3)$. Let $G'$ be the resulting graph. By gluing safeness of the applied reduction rules and Lemma~\ref{lem:bk_superset}, the boundaried graph $G'_B$ is gluing equivalent to $G_B$ with respect to the problem \vertexcover.
    \end{proof}
    
    \section{Conclusion}\label{section:conclusion}
    
    Boundaried kernelization is a recently introduced model for efficient local preprocessing due to Antipov and Kratsch~\cite{AntipovKratsch2025-BoundariedKernelization}. That work gave polynomial boundaried kernelizations for several pure graphs problems, all based on local reduction rules, as well as several unconditional lower bounds. We showed that also global tools like the matroid-based approach of Kratsch and Wahlstr\"om~\cite{KratschWahlstrom2020-RepresentativeSetsIrrelevantVertices} can be leveraged for boundaried kernelization. Moreover, this required to extend the underlying definitions to work for annotated graphs, e.g., graphs with a distinguished set of terminal vertices, including a natural generalization of gluing. 
    
    We think that this should also be possible for other problems covered by these tools, e.g., \parameterizedproblem{Vertex Cover}{$k-LP$}, \parameterizedproblem{Subset Feedback Vertex Set}{$k$}, and \parameterizedproblem{Group Feedback Vertex Set}{$k$}, possibly with some restriction as was necessary for \smultiwaycut. Similarly, using a natural notion of boundaried formula, it would be interesting to get, if possible, a polynomial boundaried kernelization for \parameterizedproblem{Almost 2-SAT}{$k$}.
    
    \bibliographystyle{abbrv}
    \bibliography{bib}
\end{document}